\documentclass[aps,prb,amsmath,amssymb,amsfonts,superscriptaddress,nofootinbib,floatfix]{revtex4} 
\pdfoutput=1
\usepackage{amsmath,amsthm,amsfonts,amssymb}
\usepackage{graphicx}
\usepackage{hyperref}
\usepackage{cleveref}

\usepackage{color}

\newtheorem{theorem}{Theorem}

\newtheorem{lemma}{Lemma}

\DeclareMathAlphabet{\mathpzc}{OT1}{pzc}{m}{it}

\newcommand{\tr}{{\rm tr}}
\newcommand{\be}{\begin{equation}}
\newcommand{\ee}{\end{equation}}
\newcommand{\Wtot}{W_{\rm total}}
\newcommand{\psitf}{\Psi_{\rm tf}}
\newcommand{\dhilb}{D_{\rm qudit}}
\newcommand{\hfk}{h^{fk}}

\begin{document}

\title{On The Complexity of Redundancy-Free Quantum Hamiltonians}

\author{Matthew B.~Hastings}
\author{Alexander Schmidhuber}

\begin{abstract}
We consider a class of \emph{redudancy-free} Hamiltonians, 
which are those where the trace vanishes for any product of Hamiltonian terms in which at least one term appears an odd number of times.
This includes Hamiltonians which are sums of products of Paulis with no relations between them.
These Hamiltonians arise naturally in Hamiltonian Decoded Quantum Interferometry \cite{schmidhuber2025}, which separates the hardness of preparing quantum states into two steps: a classical decoding step, and a step preparing the thermofield double state of such a redundancy-free Hamiltonian.

We focus on the
hardness of these Hamiltonians, both for intrinsic interest and to understand the complexity of
the state preparation step in Hamiltonian DQI. We present several results showing that the problem becomes
easy at quadratically lower temperature than for a general Hamiltonian, with the complexity at given inverse temperature
$\beta$ depending on $\beta^2 d$, where $d$ is the degree of the anticommutation graph.  We give an efficient classical algorithm for small $\beta^2 d$ to approximate the partition function, while for large $\beta^2 d$ we show that approximating the partition function is NP-hard. To this end, we introduce what we call an \emph{anticommutation glass}, where frustration arises purely from anticommutation relations. 
We give a subexponential time quantum algorithm to prepare the thermofield state for small $\beta^2 d$, and present some results toward a polynomial time algorithm based on the Feiguin-Klich Hamiltonian. As a result of possible independent interest, we show that this Hamiltonian leads to a polynomial time algorithm for preparing thermofield double states of general Hamiltonians on an interaction graph of degree $d$ for $\beta\lesssim 1/d$. Finally, we show that anticommutation alone suffices to mimic the hardness of general quantum Hamiltonians: estimating the ground state energy of redundancy-free Hamiltonians is QMA-complete.

\end{abstract}
\maketitle

\section{Introduction and Definitions}
We consider a class of Hamiltonians of the form
\be
H=\sum_i J_i O_i,
\label{eq:one}
\ee
where the $J_i$ are scalars and $\Vert O_i \Vert \leq 1$,
such that the trace vanishes for any product of terms in the Hamiltonian in which some term appears an odd number of times.
We call these Hamiltonians \emph{redundancy-free Hamiltonians}\footnote{We choose to call these Hamiltonians redundancy-free instead of relation-free since individual terms might still satisfy certain relations. For instance, Pauli operators always square to the identity. }. 
One important special case is where
the $O_i$ are products of Paulis, but with no relations between the $O_i$, meaning that
no nontrivial product of the $O_i$ is equal to the identity; we call these Pauli redudancy-free.
Another interesting case is to take a Hamitonian which is an arbitrary sum of terms $H=\sum_i J_i O_i$ without imposing any requirement on the traces and define a new Hamiltonian $\sum_i J_i Z_{{\rm anc},i} O_i$, where $Z_{{\rm anc},i}$ is a Pauli $Z$ operator acting on an extra ancilla qubit, one for each term in the Hamiltonian.

A primary motivation for studying redundancy-free Hamiltonians is the quantum algorithm Hamiltonian Decoded Quantum Interferometry (HDQI) introduced in Ref.~\onlinecite{schmidhuber2025}. For a generic Hamiltonian, HDQI prepares its Gibbs state (and other polynomial filters) by first preparing a ``pilot state" that encodes the anticommutation structure of $H$, and subsequently solving a classical decoding problem defined by the the relations between the terms in the Hamiltonian. Understanding the computational complexity of the pilot state preparation was left as a major open question in Ref.~\onlinecite{schmidhuber2025}, and the lack of general techniques for this problem  limited the applicability of HDQI outside of the commuting or nearly-commuting setting. 

We observe that the pilot state in HDQI is equivalent to the thermofield double (a purification of the Gibbs state) of a certain redundancy-free relaxation of the Hamiltonian. HDQI therefore reduces the preparation of the thermofield double of a general Hamiltonian to that of a redundancy-free one by classically decoding its relations. The two steps carry two different 
sources of hardness: anticommutation relations versus classical redundancies between the terms of the Hamiltonian. 

In this paper, we characterize the hardness of the first step. Beyond the algorithmic implications for HDQI, an additional motivation is to understand the physical properties (such as frustration, hysteresis, and entanglement) that arise from anticommutation relations alone, even in the absence of any redundancies.

In the Pauli redundancy-free case,
 a set of operators $O_i$ as in \eqref{eq:one} is characterized by an \emph{anticommutation graph} $G$.  This is a graph with one vertex per operator $O_i$, labeling the vertex by the index $i$, and with an edge between vertices $i,j$ if $\{O_i,O_j\}=0$ and with no edge if instead $[O_i,O_j]=0$.
Given an anti-commutation graph with $n$ vertices, one possible representation of the operators $O_i$ on $n$ qubits is given by taking
\be
\label{Oirep}
O_i=X_i \prod_{j<i, (i,j)\in E} Z_j,
\ee
where $(i,j) \in E$ indicates that the edge from $i$ to $j$ is in the edgeset $E$ of the graph. 
For more general redundancy-free Hamiltonians, we instead have a \emph{non-commutation graph}, so that operators commute unless they are neighbors, in which case no specific commutation relation is imposed. In the Pauli setting, we call \eqref{Oirep} the canonical redundancy-free representation of a graph, and, for a general Pauli Hamiltonian $H$, we call the canonical redundancy-free representation of its anticommutation graph the \textit{redundancy-free relaxation} of $H$.  

For the rest of the paper, we will, for simplicity and for unity of presentation,  consider the Pauli case and refer to anticommutation graphs.  However, all our results about complexity of evaluating the partition function will apply in the more general case. 

For normalization, we fix $|J_i|\leq 1$ for all $i$.  (Indeed, due to the absence of relations, the sign of each $J_i$ is irrelevant for thermal properties of the systems and for any $i$, in the representation above the operator $Z_i$ anticommutes with $O_i$ and commutes with $O_j$ for $j\neq i$).

We are interested in the complexity of computing properties of these Hamiltonans, such as approximating the partition function.  Of course, if all $O_i$ commute, then the problem is trivial, reducing to $n$ decoupled degrees of freedom. Accordingly, the ``pilot state'' of standard DQI~\cite{Jordan2025DQI} is a simple superposition of Dicke states.

We consider the complexity of the problem as a function of the degree $d$ of the anticommutation graph. 
Define
\be
\kappa=\beta^2 d.
\ee
The central theme of our results is that the problem is ``easy" at sufficiently small $\kappa$, i.e., for $\beta$ sufficiently small compared to $1/\sqrt{d}$.
This contrasts with general Hamiltonians, where the problem is only ``easy" for $\beta$ sufficiently small compared to $1/d$.
This gives evidence that Hamiltonian Decoded Quantum Interferometry may simplify the preparation of thermal states of certain Hamiltonians, as this redundancy-free Hamiltonians may be ``easy" at $\beta$ where the decoding step is efficient. 

 In \cref{sec:polymer}, we show that for sufficiently small $\kappa$, a polymer expansion for the free energy converges, giving an efficient classical algorithm to approximate the partition function; the time is polynomial in both $n$ and the error $\epsilon$.
 Conversely, for $\kappa$ sufficiently large, we show in \cref{sec:hardness} that the problem is NP-hard. We start by giving physics arguments and numerical simulations, arguing that we can effectively encode the partition function of an arbitrary Ising Hamiltonian, with both ferromagnetic and antiferromagnetic ferms, into one of these redundancy-free Hamiltonians.  This is perhaps slightly surprising: such an Ising Hamiltonian is a sum of terms $Z_i$ or $Z_i Z_j$, with some coefficients, and it \emph{has redundancy} among the terms in the Hamiltonian, but here we encode it into a problem without redundancy. We then make this intuition precise by giving a formal proof of hardness. Moreover, we show that the ground state energy estimation problem of redundancy-free Hamiltonians remains QMA-hard.

In \cref{sec:QSP}, we consider the problem of preparing a thermofield double state on a quantum computer, and give a subexponential algorithm to do this for sufficiently small $\kappa$.  In \cref{sec:beyond}, we present some results towards improving the runtime of the algorithm to polynomial.  While we do not succeed in doing this for redundancy-free Hamiltonians at $\beta \lesssim 1/\sqrt{d}$, we do
give an
algorithm for efficient preparation of thermofield double states of general Hamiltonians on an interaction graph of degree $d$ for $\beta \lesssim 1/d$.  This algorithm can be compared to
Ref.~\onlinecite{rouze2025efficient} which also works for $\beta \lesssim 1/d$; however, our algorithm may have better scaling for high-dimensional lattices as the allowed $\beta$ in Ref.~\onlinecite{rouze2025efficient} depends exponentially on lattice dimension while we make no assumption on lattice dimension.
See also Ref.~\onlinecite{bakshi2024high}.

Note on use of AI: most of the quantum Monte Carlo code was written by AI, and AI was used to find some references on the polymer expansion. AI was also used to help with language editing and proofreading in certain sections. No other use was made of AI in this project.

\section{Hamiltonian Decoded Quantum Interferometry}
\label{sec:hdqi_background}
In this section, we briefly explain how redundancy-free Hamiltonians arise in Hamiltonian Decoded Quantum Interferometry.
Decoded quantum interferometry (DQI) prepares a state whose amplitudes are a polynomial of a classical cost function by coherently running a classical decoder~\cite{Jordan2025DQI}. Hamiltonian DQI (HDQI) extends this idea to polynomials of quantum Hamiltonians~\cite{schmidhuber2025}. The quantum input to this procedure is the so-called pilot state, which captures the anticommutation structure of the Hamiltonian and for thermal states is Clifford-equivalent to the thermofield double (TFD) of its redundancy-free relaxation. The decoder then removes the extra register that distinguishes products related in the physical Hamiltonian. For more details on HDQI, consult Ref.~\onlinecite{schmidhuber2025}. For simplicity, we focus on the Pauli setting throughout this section. 

Write the Hamiltonian on $n$ qubits as
\begin{equation}
\label{eq:hdqi_hamiltonian}
H_v=\sum_{i=1}^m J_i v_iP_i,
\qquad J_i\geq0,\qquad v_i\in\{\pm1\},
\end{equation}
where the $P_i$ are Hermitian Pauli operators. Let $b_i=(x_i,z_i)\in\mathbb F_2^{2n}$ record the $X$ and $Z$ exponents of $P_i$, ignoring its scalar phase. Multiplication adds these labels. Set $B=(b_1\ \cdots\ b_m)$. The \emph{relation code} is
\begin{equation}
\label{eq:hdqi_relation_code}
\mathcal C=\ker B.
\end{equation}
Indeed, the ordered product $P(y)=P_1^{y_1}\cdots P_m^{y_m}$ is scalar exactly when $By=0$. Decoding means recovering $y$ from its syndrome $By$.

Let $A$ be the anticommutation matrix, with $P_iP_j=(-1)^{A_{ij}}P_jP_i$. On $m$ qubits take the canonical redundancy-free generators and Hamiltonian
\begin{equation}
\label{eq:hdqi_relaxation}
\widehat P_i=X_i\prod_{j<i}Z_j^{A_{ij}},
\qquad
\widehat H=\sum_iJ_i\widehat P_i.
\end{equation}
Their ordered products $\widehat P(y)=\widehat P_1^{y_1}\cdots\widehat P_m^{y_m}$ are orthonormal for normalized trace. Expand
\begin{equation}
\label{eq:hdqi_thermal_expansion}
e^{-\beta\widehat H/2}=\sum_y a_y\widehat P(y),
\qquad
\widehat Z=\tr(e^{-\beta\widehat H})=\sum_y|a_y|^2.
\end{equation}
Reordering a word uses only the pairwise commutation signs, which are the same for the physical and relaxed generators. The same coefficients therefore give
\begin{equation}
\label{eq:hdqi_physical_expansion}
e^{-\beta H_v/2}=\sum_y a_yv^yP(y),
\qquad v^y=\prod_i v_i^{y_i}.
\end{equation}
The expansion is worked out in Appendix~\ref{app:hdqi_expansion}.

For an operator $Q$ on $k$ qubits, write $|Q\rangle=(Q\otimes I)\Psi_{\rm ME}$, where $\Psi_{\rm ME}=2^{-k/2}\sum_z|z\rangle_L|z\rangle_R$. Then $\langle Q|R\rangle=\tr(Q^\dagger R)$. The redundancy-free TFD is $|\widehat\psi_\beta\rangle=|e^{-\beta\widehat H/2}\rangle/\sqrt{\widehat Z}$, while the HDQI pilot and its error distribution are
\begin{equation}
\label{eq:hdqi_pilot}
|\pi_\beta\rangle=\frac1{\sqrt{\widehat Z}}\sum_y a_y|y\rangle,
\qquad
\mu_\beta(y)=\frac{|a_y|^2}{\widehat Z}.
\end{equation}
A Clifford circuit converts one state into the other. To see this, write
\begin{equation}
\label{eq:hdqi_vectorized_product}
|\widehat P(y)\rangle
=2^{-m/2}\sum_z
(-1)^{\sum_{j<i}A_{ij}y_i z_j}
|y+z\rangle_L|z\rangle_R.
\end{equation}
CNOTs from $R_i$ to $L_i$ put $y$ in the left register. A controlled-$Z$ between $L_i$ and $R_j$ for each $j<i$ with $A_{ij}=1$ removes the phase, and Hadamards on the right register give $|y\rangle|0^m\rangle$. Calling this circuit $W_A$, we have
\begin{equation}
\label{eq:hdqi_clifford}
W_A|\widehat\psi_\beta\rangle
=|\pi_\beta\rangle|0^m\rangle.
\end{equation}
The same equivalence holds for any redundancy-free Pauli realization of the generators, as shown in Appendix~\ref{app:hdqi_expansion}.

The remaining algorithimic steps of HDQI are as follows. Starting with the pilot, apply the phase $v^y$ and the controlled ordered product $P(y)$ to one half of $n$ Bell pairs. A Bell basis change exposes the label $By$. We compute $\operatorname{Dec}(By)$, add it to the pilot register, and undo the decoder workspace and the Bell basis change. This gives
\begin{equation}
\label{eq:hdqi_coherent_output}
|\Psi_v\rangle
=\frac1{\sqrt{\widehat Z}}\sum_y a_yv^y
|y+\operatorname{Dec}(By)\rangle\,|P(y)\rangle.
\end{equation}
Correctly decoded labels all leave the first register in $|0^m\rangle$, so their contributions interfere. We discard the first register at the end. After preparing the redundancy-free TFD, the cost is polynomial in $m+n$ plus the cost of running the decoder reversibly.

When the terms commute, $\widehat H=\sum_iJ_iX_i$. The pilot is a product state, with its $i$th factor proportional to $\cosh(\beta J_i/2)|0\rangle-\sinh(\beta J_i/2)|1\rangle$. Its bits are independent and have error probabilities
\begin{equation}
\label{eq:hdqi_bsc}
p_i=\frac{\sinh^2(\beta J_i/2)}{\cosh(\beta J_i)}
=\frac{1-\operatorname{sech}(\beta J_i)}2
=O(\beta^2 J_i^2).
\end{equation}
For equal coupling magnitudes this is a binary symmetric channel. Hence HDQI applied to commuting Hamiltonian gives a  conceptually simple picture: the pilot state is always efficiently preparable since there are no anti-commutation relations, and the decoding problem to reach inverse temperature $\beta$ is that of the binary symmetric channel (BSC). As $\beta \to \infty$ so physical temperature approaches 0, the BSC noise rate approaches the information-theoretic limit $p_i = 1/2$.

Noncommuting terms introduce correlations such that the pilot and error distribution  do not have this simple product structure, but (as shown in Appendix~\ref{app:hdqi_weights}) the error distribution $\mu_\beta$ still satisfies
\begin{equation}
\label{eq:hdqi_local_stochastic}
\Pr_{Y\sim\mu_\beta}[S\subseteq\operatorname{supp}Y]
\leq\prod_{i\in S}\left(\frac{\beta J_i}{2}\right)^2
\qquad(S\subseteq[m]).
\end{equation}
For small $\beta$, both these results achieve quadratically better scaling in $\beta$ than the generic scheme first introduced in Hamiltonian DQI, which approximates a Gibbs state by choosing a polynomial approximation to the exponential function instead of working with the exact expansion \eqref{eq:hdqi_physical_expansion}.

 We focus on Gibbs states and exponential functions in this document for physical and conceptual simplicity; every result here directly extends to arbitrary polynomials as studied in HDQI. For a degree-$\ell$ polynomial, the pilot has a sharp weight cutoff: $a_y=0$ for $|y|>\ell$. The thermal pilot instead has a small tail at high weights. To show that this does not impact the applicability of HDQI, at least in the setting of randomly signed couplings, we prove in Appendix~\ref{app:hdqi_robustness} the following robustness theorem. In words, for a Hamiltonian
\begin{equation}
H_v=\sum_{i=1}^m v_iJ_iP_i
\end{equation}
with uniformly random signs $v$, HDQI prepares the thermofield double of $H_v$ to high fidelity for all but a $\delta$-fraction of sign choices, provided (1) the  thermofield double of the redundancy-free relaxation is prepared with high fidelity and (2) the decoder succeeds with high probability. The two sources of error simply add: imperfections in the pilot are not amplified by decoding failures, and the decoder only needs to succeed whp on average, not in the worst case.

\begin{theorem}[Robustness of HDQI]
\label{thm:hdqi_robustness}
Fix $\beta\geq0$, the operators $P_i$, the magnitudes $J_i$, and a decoder $\operatorname{Dec}$, and let
\[
r=\Pr_{Y\sim\mu_\beta}[\operatorname{Dec}(BY)\neq Y].
\]
Choose the signs $v_i$ independently and uniformly. With the exact pilot, the physical output $\rho_v$ of \eqref{eq:hdqi_coherent_output} satisfies
\begin{equation}
\label{eq:hdqi_average_infidelity}
\mathbb E_v\bigl[1-\langle\psi_{\beta,v}|\rho_v|\psi_{\beta,v}\rangle\bigr]
\leq2r,
\qquad
|\psi_{\beta,v}\rangle=
\frac{|e^{-\beta H_v/2}\rangle}{\sqrt{\tr(e^{-\beta H_v})}}.
\end{equation}
If the prepared redundancy-free TFD has squared fidelity at least $1-\varepsilon_{\rm p}$, the resulting output $\widetilde\rho_v$ obeys
\begin{equation}
\label{eq:hdqi_robustness}
D_{\rm tr}\!\left(\widetilde\rho_v,
|\psi_{\beta,v}\rangle\langle\psi_{\beta,v}|\right)
\leq\sqrt{\varepsilon_{\rm p}}+\sqrt{\frac{2r}{\delta}}
\end{equation}
with probability at least $1-\delta$ over the signs, for every $0<\delta<1$. Here $D_{\rm tr}(\rho,\sigma)=\tfrac12\operatorname{Tr}|\rho-\sigma|$.
\end{theorem}

\section{Polymer Expansion at Small $\kappa$}
\label{sec:polymer}
The partition function $Z=\tr(\exp[\beta H])$ can be expanded as a power series
$$
Z=1+\sum_{k\geq 1} \frac{\beta^k}{k!} \tr(H^k).
$$
In turn, $\tr(H^k)$ can be written as a sum of $n^k$ terms,
$$\tr(H^k)=\sum_{i_1,\ldots,i_k} \prod_{a=1}^k J_{i_a} \tr(O_{i_1} O_{i_2} \ldots O_{i_k}).$$
The trace vanishes unless each operator $O_i$ appears an even number of times, in which case the trace is equal to $\pm 1$ for the Pauli case and is bounded in absolute value by $1$ in general.
Any given choice of $i_1,\ldots,i_k$ induces some subgraph $G_{i_1,\ldots,i_k}$ of the anticommutation graph $G$.  This subgraph has a vertex $i$ iff $i\in \{i_1,\ldots,i_k\}$, and has an edge $(i,j)$ between two such vertices iff $\{O_i,O_j\}=0$.

Write the vertex set $V_{i_1,\ldots,i_k}$ of this subgraph as a union of connected components; here connectivity is with respect to this subgraph $G_{i_1,\ldots,i_k}$.

\subsection{The Polymer Partition Function}
We now introduce the polymer partition function.

Consider any given sequence $i_1,\ldots,i_k$ where the vertex set has $C$ connected components.  Consider any other
sequence $j_1,\ldots,j_k$ that can be obtained from $i_1,\ldots,i_k$ by a permutation that \emph{preserves the order in each connected component}.
For example, if $C=2$ and $i=1,i=2$ are in one component and $i=3,i=4$ are in another component, then $1,2,1,2,3,4,3,4$ can be permuted to $1,3,2,4,1,3,2,4$ as the order $1,2,1,2$ and $3,4,3,4$ is the same in both.
Then, both sequences give the same term:  $\beta^k \prod_{a=1}^k J_{i_a} \tr(O_{i_1} O_{i_2} \ldots O_{i_k})=\beta^k \prod_{a=1}^k J_{j_a} \tr(O_{j_1} O_{j_2} \ldots O_{j_k})$.
Let's count the number of sequences that can be obtained from a given $i_1,\ldots,i_k$ by such a permutation.
Suppose the sequence $i_1,\ldots,i_k$ has $n_c$ entries in the $c$-th connected component, counting multiplicity.  For example, the sequence $1,2,1,2,3,4,3,4$ has $n_1=n_2=4$.  Given $i_1,\ldots,i_k$, we can specify any sequence $j_1,\ldots,j_k$ simply by indicating which component each term of $j_1,\ldots,j_k$ is from.  That is, by giving a sequence of $k$ numbers, each from $1,\ldots,C$, with $n_1$ appearances of $1$, $n_2$ appearances of $2$, and so on.
The number of such sequences is equal to $k!/(\prod_{c=1}^C n_c!)$.

Thus, for any given sequence $i_{1,1},\ldots,i_{1,n_1},i_{2,1},\ldots,i_{2,n_2},\ldots,i_{C,1},\ldots,i_{C,n_C}$, written so that $i_{a,1},\ldots,i_{a,n_a}$ are all in the $a$-th connected component, if we sum over all sequences related by such a permutation to that given sequence,
the total contribution to $Z$ is equal to
$$\prod_{c=1}^C \Bigl( \frac{1}{n_c!} \prod_{a=1}^{n_c} (\beta J_{i_{c,a}}) \tr(O_{i_{c,1}} O_{i_{c,2}} \ldots O_{i_{c,k}})\Bigr).$$

We can then express the sum over all contributions to $Z$ by summing over the number of connected components, summing over choices of connected components, summing over sequences in each connected component, and for each choice of a sequence in each connected component, using the above expression for the contribution to $Z$.
We can replace the sum over sequences in each connected component by a sum over multisets (i.e., sets of vertices with multiplicities), followed by a sum over sequences in which each element of that multiset appears a number of times equal to its multiplicity.  We say that such sequences are \emph{compatible} with the given multiset.
Formally, let $\Gamma$ be the set of all finite multisets $\gamma$ such that the subgraph induced by $\gamma$ (i.e., the subgraph containing the vertices which appear at least once in that multiset)  has exactly one connected component.
Let the weight $w(\gamma)$, where $\gamma$ is a given multiset, be given by
$$w(\gamma)=\frac{1}{m!} \; \sum_{\stackrel{i_1,\ldots,i_m}{ {\rm compatible\, with\,} \gamma}}\;  \prod_{a=1}^{m} (\beta J_{i_a}) \tr(O_{i_{1}} \ldots O_{i_m}),$$
where $m$ is the total multiplicity of the multiset
Note that we can bound the weight $w(\gamma)$ by the product $\prod_i (\beta J_{i})^{m_i}$, where $m_i$ is the multiplicity of $i$ in the given multiset, since there are at most $m!$ different sequences compatible with the given multiset.
Since $|J_i|\leq 1$, this bounds $w(\gamma) \leq \beta^{\sum_i m_i}$.

Then,
\be
\label{Zpolymer}
Z=\sum_{\Gamma'\subset \Gamma} \bigl( \prod_{\gamma\in \Gamma'} w(\gamma)\bigr) \bigl(\prod_{\{\gamma,\gamma'\}\subset \Gamma'}\delta(\gamma,\gamma'),
\ee
where here $\delta(\gamma,\gamma')$ enforces the requirement that different components must be non-adjacent to each other.  We set $\delta(\gamma,\gamma')=0$ if some vertex in $\gamma$ is equal to or adjacent to some vertex in $\gamma'$.
Otherwise we set it equal to $1$.
We include the term $\Gamma'=\emptyset$ in this sum, giving a contribution of $1$ to the partition function.

\Cref{Zpolymer} is in the form of a polymer partition function.  Each $\gamma$ is called a \emph{polymer}.
A theory has been developed in statistical mechanics studying these partition functions.  See for example chapter 5 in Ref.~\onlinecite{friedli_velenik_2017} for a review.

Before proceeding and using this theory,
it is convenient to further group all multisets having the same support into one polymer.  From here on, $\gamma$ will label then a set of sites, rather than a multiset.  Summing over multisets having the given support, summing $m_i=2,4,\ldots$ for each $i$,
we bound
\be
w(\gamma) \geq \bigl(\frac{ \beta^2}{1-\beta^2}\bigr)^{|\gamma|},
\ee
where $|\gamma|$ is the cardinality of the support.  Note here that we use that all $m_i$ must be even or else the trace $\tr(O_{i_{1}} \ldots O_{i_m})$.  \emph{This requirement that each $m_i$ be even is the key idea that leads to convergence at $\beta\sim 1/\sqrt{d}$}.

The key idea in this theory is roughly as follows.  If the weights are small, one might hope to develop a good approximation theory, using a small number of terms in the sum, i.e., sets $\Gamma'$ which have small cardinality.  However, if we directly study the partition function, the small weights are counteracted by the possible large number of polymers.  Just to give an example, suppose we take all operators $O_i$ to commute, i.e., $O_i=Z_i$, and let all $J=1$.  Then, the simplest polymer repeats some given index twice, and has weight $\beta^2$.  For small $\beta$, this weight is small, but the sum over all terms with $|\Gamma'|=1$ is equal to $n\beta^2$ and hence is not small for large enough $n$, and then the next term with two polymers is $\approx n^2 \beta^4/2$ and may be even larger.
However, if we study the logarithm of the partition function, these problems are resolved in many cases, and the series converges rapidly if the weights are small enough and if each polymer interacts with few enough other polymers.  Here the interaction between polymers is due to the requirement that the different components be non-adjacent to each other, i.e., it is caused by non-zero $\zeta(\gamma,\gamma')$, where
$\zeta(\gamma,\gamma')\equiv 1-\delta(\gamma,\gamma')$.

\subsection{The Series for the Log Partition Function and a Convergence Criterion}
It is proven that, under a certain convergence condition $(*)$ given below,
we can write $Z=\exp[F]$, where $F$ has the power series expansion
\be
\label{Fexpansion}
F=\sum_{m\geq 1} \sum_{\gamma_1,\ldots \gamma_m} \phi(\gamma_1,\ldots,\gamma_m) \prod_i w(\gamma_i),
\ee
where $\phi(\gamma_1,\ldots,\gamma_m)$ is defined as follows.  If $m=1$, $\phi(\gamma)=1$.  If $m\geq 1$, define $\phi(\cdot)$ as follows.  Let $\sum_{G_{{\rm conn},m}}$ denote the sum over graphs $G_{{\rm conn},m}$ which are connected and which have $m$ vertices, labeled $1,\ldots,m$.  Then for $m\geq 2$,
$$\phi(\gamma_1,\ldots,\gamma_m)=\frac{1}{m!} \sum_{G_{{\rm conn},m}} \prod_{(i,j)\in G_{{\rm conn},m}} \zeta(\gamma_i,\gamma_j).$$

The required convergence condition $(*)$ is that there exist some positive real number $a(\gamma)$ for each $\gamma$ so
that
\be
\label{convcon}
\sum_{\gamma' \neq \gamma} |w(\gamma')| e^{a(\gamma')} \zeta(\gamma,\gamma')\leq a(\gamma).
\ee
Under this assumption, it is proven (see theorem 5.4 in \cite{friedli_velenik_2017}) that the expansion \cref{Fexpansion} converges absolutely, and that $Z=\exp[F]$.
Further, it is proven that 
\be
1+\sum_{k\geq 2} k \sum_{\gamma_2,\ldots,\gamma_k} |\phi(\gamma_1,\gamma_2,\ldots,\gamma_k)| \prod_{j=2}^k |w(\gamma_j)| \leq e^{a(\gamma_1)}
\ee
for all $\gamma_1$, and hence
\be
\label{oneterm}
w(\gamma_1)+\sum_{k\geq 2} k \sum_{\gamma_2,\ldots,\gamma_k} |\phi(\gamma_1,\gamma_2,\ldots,\gamma_k)| \prod_{j=1}^k |w(\gamma_j)| \leq e^{a(\gamma_1)} w(\gamma_1)
\ee
\cref{oneterm} bounds the sum of absolute values of terms in \cref{Fexpansion} which include polymer $\gamma_1$.

Indeed, \cref{convcon} does hold for sufficiently small $\kappa$ if we pick
\be
a(\gamma)=|\gamma|.
\ee
To show this,
we need a bound on the number of connected subsets of the anticommutation graph.  
Such subsets are called lattice animals.  Using standard bounds, the number of polymers $\gamma$ with cardinality $|\gamma|=W$ which include some given 
vertex of the anticommutation graph
is bounded by
$C (e d)^{W-1}$ for some constant $C$.
Any polymer $\gamma'$ which has non-vanishing $\zeta(\gamma,\gamma')$ must include some vertex in $\gamma$ or which neighbors $\gamma$.  There are at most $(d+1) |\gamma|$ such vertices.
Hence, the number of polymers $\gamma'$ with non-vanishing $\zeta(\gamma,\gamma')$ and
with given $|\gamma'|=W$ is bounded by $C (d+1) (e d)^{W-1} |\gamma|
\leq C' (e d)^{W}$ for some other constant $C'$,  while the weight of such polymers is bounded by $(\beta/(1-\beta^2))^W$, and of course
$\exp(a(\gamma'))=\exp(W)$.
So, the sum over $\gamma'$ in \cref{convcon} is
bounded by
$$\sum_{W\geq 1}C' \bigl(\frac{e^2d \beta^2}{1-\beta^2}\bigr)^{W}
 |\gamma|,
 $$
 and for
 sufficiently small $\kappa$ this inequality holds.
 Here we are assuming that $d\geq 1$ so that small $\kappa$ also implies $\beta^2<1$.
 
 Let us know see how rapidly \cref{Fexpansion} converges.
 Consider the sum of all terms in \cref{Fexpansion} which include some given polymer $\gamma$ and where the total cardinality of the polymers appearing in that is equal to some $\Wtot$.
Let $\kappa_0$ be such that the inequality holds for all $\kappa<\kappa_0$.
Then, for any given $\kappa<\kappa_0$, since the sum of these terms converges absolutely for $\kappa<\kappa_0$, indeed the sum is bounded by
$a(\gamma) w(\gamma) (\kappa/\kappa_0)^{\Wtot}$.

Hence, for any fixed ratio $\kappa/\kappa_0<1$, this sum is exponentially small in $\Wtot$, and so we may truncate up to terms of total cardinality of support $<\Wtot$ with exponentially small error.

In particular, the derivative of $F$ with respect to any coupling constant can be expressed as a sum of terms including polymers which include the given vertex, and so if we truncate the series for this at total cardinality of support $<\Wtot$ the error $\epsilon$ is bounded by
$C (\kappa/\kappa_0)^{\Wtot}$ for some constant $C$.
Similarly, the error $\epsilon$ in $F$ if we truncate the series at total cardinality of support $<\Wtot$ is bounded by
$C n (\kappa/\kappa_0)^{\Wtot}$.

\subsection{Complexity of Calculating The Truncation}
\label{compcalc}
We now consider the complexity of computing this truncated series, in particular the sum of all terms appearing in $F$ up to $\Wtot$.
First consider the complexity of computing the weight $w(\gamma)$ of some $\gamma$ with $|\gamma|=W$.
Recall that $w(\gamma)$
equals the sum $\beta^k \prod_{a=1}^k J_{i_a} \tr(O_{i_1} O_{i_2} \ldots O_{i_k})$
over \emph{all} sequences $i_1,\ldots,i_k$ such that each vertex in $\gamma$ appears at least once in the sequence and no vertices outside $\gamma$ appear in the sequence.  We could approximate $w(\gamma)$ by explicitly summing over such sequences, but the number of sequences grows factorially in $|\gamma|$.

A more efficient technique is as follows.   Identify the label $\gamma$ of a polymer with sets of vertices.  For any such set $\gamma$, let $$Z(\gamma,\beta)=\tr(\exp[\beta H(\gamma)]),$$ where $H(\gamma)=\sum_{i\in \gamma} J_i O_i$.
For any $\gamma$, in the Pauli case we can represent the operators $O_i$ for $i\in \gamma$ using $|\gamma|$ qubits and obtain the correct commutation/anticommutation relations; in the general case we must assume some representation of the operators on a bounded number of qudits exists.  Using this representation, we can compute $Z(\gamma,\beta)$ in time $\exp(O(|\gamma|))$ and up to an accuracy $\epsilon$, with the time scaling polylogarithmically in the error.

Of course, the series expansion of $Z(\gamma,\beta)$ includes terms where not every $i\in \gamma$ appears, meaning that the corresponding operator $O_i$ is not in the trace in the given term.
However, we can solve this using an inclusion-exclusion method.
We have
\be
w(\gamma)=\sum_{\gamma'\subset \gamma} (-1)^{|\gamma|-|\gamma'|} Z(\gamma',\beta).
\ee
Since there are only $2^{|\gamma|}$ subsets $\gamma'$ of $\gamma$, the total time to compute $w(\gamma)$ is 
$\exp(O(|\gamma|))$, again up to a polylogarithm in the error.

We now count the number of terms that appear in this truncated series for $F$.  Any term that appears involves some number of polymers, $\gamma_1,\ldots,\gamma_m$.
These polymers must be connected with each other in the following sense: there is some permutation $\pi$ so that $\gamma_{\pi(i)}$ neighbors $\gamma_{\pi(i+1)}$ for $i=1,\ldots,m-1$, where we say that two polymers neighbor each other if they intersect on a vertex or if some vertex in one neighbors some vertex in another.
Thus, these collection of polymers can be regarded as a lattice animal on an enlarged graph which consists of $m$ copies of the anticommutation graph where we additional add edges between each vertex $v$ in the $i$-th copy with the same vertex $v$ in the $(i+1)$-st copy as well as to all neighbors of $v$ in the $(i+1)$-st copy.  The vertices in the lattice on the enlarged graph in the $i$-th copy are simply those in $\gamma_{\pi(i)}$.
The degree of this enlarged graph is bounded by $3d+2$.
So, we can bound the number of such animals by
$n C (C' d)^W$
for some constants $C,C'$.

Truncating at $W<\Wtot$, and using the exponential complexity of computing each $w(\gamma)$, the complexity of evaluation is bounded by
$$O(n (C d )^{\Wtot})=O(n \exp(O(\log(d)) \Wtot)),$$ for some other constant $C$.

To obtain error $\epsilon$ in $F$ we need
$\Wtot\sim \log(n/\epsilon)/\log(\kappa^{-1}).$
Hence, the complexity is
$$O(n (n/\epsilon)^{O(\log(d)/\log(\kappa^{-1}))}).$$

For any fixed $d,\kappa$, this time is polynomial.

\section{NP-hardness at Large $\kappa$}
\label{sec:hardness}

We now show that approximating the partition function for sufficiently large $\kappa$ is NP-hard.  

We first construct a simple system that in some sense acts as a classical Ising spin: the thermal state of this Hamiltonian has two distinct ``basins".  In a particular representation of the operators $O_i$ that we will use (different from the representation in \cref{Oirep}), these basins will be either $X$ or $Z$ polarized for the spins.  These basins act as the state of this classical Ising spin.  By adjusting the strength of certain coupling terms in the Hamiltonian, we can make one or the other of the basins preferred, acting a magnetic field on this classical Ising spin. For the formal hardness proof, however, it suffices to use equal strengths and no magnetic fields.  Finally, we then show that by adding appropriate coupling terms between a pair of these systems we can effectively induce ferromagnetic or anti-ferromagnetic interactions between these effective Ising spins.
We can then build a system that mimics a classical Ising model by taking many such systems, adding interactions and magnetic fields in this way, and then approximating the partition function is NP-hard; the product of $\beta$ times the effective interaction strength of this effective classical Ising model will be large, growing with $d$ so that model can be in the low temperature regime.
The important point in this construction is that \emph {the effective Ising model can have redundancies}.

The simple system we construct may have independent interest as a quantum spin glass.  It shares some relation to the quantum random energy model\cite{goldschmidt1990solvable}.
The anticommutation graph will have two anticliques, i.e., two sets of operators, where operators in a given set commute with all others in that set, while operators are roughly equally likely to commute or anticommute if they are in different sets.
Frustration and glassiness, including properties such as hysteresis, arises from the anticommutation relation.
Hence, we call it the \emph{anticommutation glass}.
Our construction of the anticommutation glass uses randomness.  However, the randomness is simply a way to obtain some simple system with some useful properties. 
In \cref{sec:NP_formal}, we use an explicit deterministic family of invertible matrices with the same two-family structure to obtain a formal NP-hardness proof. Rather than assuming a decomposition of the physical thermal state into two basins, we will make this picture precise through upper and lower bounds on the full partition function.

The reader interested primarily in the formal hardness proof should read \cref{subsec:ag,subsec:cag} for the physical motivation before proceeding to \cref{sec:NP_formal}, but may skip \cref{subsec:se,subsec:ns,subsec:tm}.
\subsection{Anticommutation Glass}
\label{subsec:ag}
The anticommutation glass Hamiltonian is as follows.  We slightly change normalization and notation in this section.  We have $n$ qubits and $2n$ operators $O_i$, $i=1,\ldots,2n$, for some integer $n$.  
The degree of the anticommutation graph for this Hamiltonian will be bounded by $n$.
When we couple several such anticommutation glasses to produce an effective classical Ising Hamiltonian, the degree of the anticommutation graph will be proportional to $n$ times the degree of the interaction graph of that Ising Hamiltonian.

The first $n$ of these operators are Pauli $X$ operators on a single qubit:
$$1\leq i \leq n  \quad \rightarrow \quad O_i=X_i.$$
On the other hand, for the remaining operators we set
$$
O_{n+i}=D_i\equiv\prod_{j=1}^n Z_j^{A_{ij}},
\qquad 1\leq i\leq n,
$$
where $A$ is an $n$-by-$n$ invertible matrix over $\mathbb F_2$.
The Hamiltonian is
\be
H_A(J)=-\sum_{i=1}^n X_i-J\sum_{i=1}^n D_i,
\ee
with $J>0$.
That is, each operator $O_i$ for $i>n$ is a product of $Z_j$ over $j$ such that $A_{i-n,j}=1$.
In a change of normalization from previously, we will allow $J>1$ so long as it is still bounded by a constant; one can overall rescale $H$ and $\beta$ by this constant to return to the previous normalization.

We  refer to operators $O_i$ with $i\leq n$ as $X$-type operators, while those with $i>n$ are $Z$-type operators.
For $i>n$ and $j\leq n$, the operators $O_i,O_j$ anticommute whenever $A_{i-n,j}=1$.

We choose the matrix $A$ randomly, but require that it be invertible.  In the numerical simulations reported below, we effectively chose such an $A$ by starting with the identity matrix and then doing a very large number (large compared to $n$) of random row operations on the matrix.

The requirement that $A$ be invertible is equivalent to the requirement that there be no relations among the $O_i$.
Choosing $A$ in this random way means that for $i\leq n$ and $j>n$, the pair of operators $O_i$ and $O_j$ are roughly equally likely to commute or anticommute with each other.

Since the matrix $A$ is invertible, we could apply a sequence of CNOT gates to the qubits so that, conjugating by these gates, $O_i=Z_{i-n}$ for $n+1\leq i \leq 2n$.  However, after applying this sequence of gates, we have $O_i = \prod_{j} X_j^{A_{ji}} $ for $1\leq i \leq n$.
We could then follow with a Hadamard on all qubits, turning Pauli $X$ to Pauli $Z$.
Thus, statistically, there is no distinction between $X$-type and $Z$-type operators, other than of course the coupling constant $J$; by applying such a sequence of CNOTs followed by Hadamards, we can interchange $X$-type and $Z$-type operators up to replacing $A$ by its transpose which is statistically the same.
The symmetry between $X$-type and $Z$-type operators is an important difference between this and the quantum random energy model.

We claim that the physical properties of the system depend on $\kappa$ as follows, with these properties observed for almost all random instances.
Consider the expectation value of a given operator, either $X$-type or $Z$-type operator.  If we \emph{only had operators of the given type and not of the other type}, then the expectation value can be determined easily: since the Hamiltonian is redundancy-free, the expectation value would be
$\tanh(\beta)$ for an $X$-type operator and $\tanh(\beta J)$ for a $Z$-type operator.
Let the \emph{ratio} be defined to be the ratio between the observed expectation value and this quantity (either $\tanh(\beta)$ or $\tanh(\beta J)$).
The ratio clearly approaches $1$ as $\kappa\rightarrow 0$; we show below that for small $\kappa$, the ratio is decreasing with $\kappa$ for both types of operators.
However, for $J>1$, for $Z$-type operators, the ratio hits a minimum at some $\kappa$ and then increases, converging to $1$ at large $\kappa$.
At the same time, for $J>1$, for $X$-type operators, the ratio is decreasing for all $\kappa$, converging to $0$ at large $\kappa$.
Note, however, that the individual expectation values of $Z$-type operators are still small compared to $1$ at large $\kappa$, so long as $\beta \ll 1$.
So, it is not that the system is fully polarized in the $Z$-direction, but rather that the $X$-type operators become negligible for the thermal properties, and we have thermal fluctuations in the $Z$-basis about this polarized state.
Conversely, for $J<1$, the same phenomenon is observed with $X$- and $Z$-type interchanged: the ratio for the $Z$-type operators converges to $0$ at large $\kappa$ while that for $X$-type converges to $1$.

Further, we observe numerically some hysteresis.  That is, starting the path-integral Monte Carlo simulations below in a state polarized in the $Z$-direction but with $J<1$ means that it sometimes gets stuck, leaving the ratio for $Z$-type operators close to $1$ at large $\kappa$ and the ratio for $X$-type operators close to $0$, even though in equilibrium the ratios should display the opposite behavior.

For this reason, we expect that, for  large $\kappa$, $J$ close to $1$, and $n$ large, there will be two different states relevant for the thermodynamics, one state stable and one state only metastable.  Call these states $X$-type or $Z$-type, where the $X$-type state is
approximately described by ignoring the $Z$-type operators and considering the thermal fluctuations for the $X$-type operators, or vice-versa.
Indeed, representing the $X$-type operators as $X_i$, the $X$-type state is simply a product state of qubits in an $X$ magnetic field at finite $\beta$. 

Thus, the $X$-type and $Z$-type states of the anticommutation glass act as the two different states of the classical Ising spin system we will build, and $J-1$ acts as an effective magnetic field on this effective classical Ising spin, choosing one or the other state.

\subsection{Coupling Anticommutation Glasses and Effective Classical Ising Model}
\label{subsec:cag}
Now consider two such anticommutation glasses.  Let us show how to couple the effective Ising spins.  We represent the first using unprimed Pauli operators such as $X_i$ and the second using primed Pauli operators, such as $X'_i$.  We will additionally add a third set of $n$ ancilla qubits written with two primes; these will be added solely to ensure that the Hamiltonian we write is redundancy-free; the only terms in the Hamiltonian involving these ancilla qubits will be those in \cref{combineacg} or \cref{combineacgAF} below.
These two anticommutation glasses
may have different coupling constants $J,J'$.

We now add an interaction term to the Hamiltonian.
For a ferromagnetic interaction, take
\be
\label{combineacg}
H_{\rm coupling}
=-\widetilde J\sum_i
\left(X_iX_i'X_i''+D_iD_i'Z_i''\right),
\ee
where $\widetilde J\geq0$.
 If both anticommutation glasses are in the same type of state (either $X$-type or $Z$-type), then the
system can increase the log of its partition function by roughly $n \ln(\cosh(\beta \tilde J))$ as one can gain from thermal fluctuations in either the $X$ or $Z$ terms in $H_{\rm coupling}$, while little increase in log partition function is possible if the glasses are in different types of state.  Thus, such a term acts a ferromagnetic coupling.

Similarly, we can induce an effective anti-ferromagnetic coupling with
\be
\label{combineacgAF}
H_{\rm coupling}
=-\widetilde J\sum_i
\left(X_iD_i'X_i''+D_iX_i'Z_i''\right).
\ee
Now an interaction term agrees with the chosen families at both endpoints precisely when the two choices are different.

The ancillas are redundancy-free even though the two interaction types use the same ancilla. In any binary relation, its $X$ and $Z$ coordinates force the coefficients of both terms using that ancilla to vanish. Repeating this leaves only relations inside the original glasses, which do not occur because their matrices are invertible.

By combining then many anticommutation glasses, we can build an arbitrary effective Ising spin system.  We keep the degree of the interaction graph for this Ising spin system bounded by a constant, so that the degree of the anticommutation graph is only increased by a constant factor.
Note that the increase in log-partition function is roughly, for small $\beta \tilde J$, equal to $\frac{1}{2}n \beta^2 \tilde J^2$, and hence is proportional to $\kappa \tilde J^2$.  Thus, while our effective Ising spin system will be effectively at nonzero temperature, at large enough $\kappa$ we can take the temperature low enough that approximating the partition function is NP-hard.

In the construction above, we have been treating the $X$-type and $Z$-type states as some effective classical states.  Likely there is some effective quantum tunneling between these states, with that tunneling becoming small at large $\kappa$.  Thus, it might be possible to encode effective quantum Ising Hamiltonians using anticommutation glasses.  We leave this as an open question.

In the next three subsections, we justify the claims about the behavior of the anticommutation glass. Then, in \cref{sec:NP_formal}, we complete our argument with a formal NP-hardness proof.

\subsection{Series Expansion}
\label{subsec:se}
First consider small $\kappa$.  We series expand the expectation values $\langle O_i \rangle$ in $\beta$.  Consider a $Z$-type operator $O_i$.
We have
\begin{eqnarray}
\langle O_i \rangle &=& \tr(\exp(-\beta H))^{-1} \tr(\exp(-\beta H) O_i) \\ \nonumber
&=& \bigl(1+\frac{1}{2}\beta^2 n (1+J^2)+\ldots\bigr)^{-1} \bigl(\beta J+\frac{n-1}{2} \beta^3 J^3+\frac{1}{3!} \beta^3 J^3 +\frac{1}{3!} \beta^3 J \sum_{j=1}^n \tr(O_i O_j O_j O_i + O_i O_j O_i O_j + O_i O_i O_j O_j)+\ldots \bigr)
\end{eqnarray}
In the second line, the term $\frac{n-1}{2} \beta^3 J^3$ arises from expanding terms in $\exp(\beta H)$ where $O_i$ appears once and some different $Z$-type operator appears twice.
The trace $\tr(O_i O_j O_j O_i + O_i O_j O_i O_j + O_i O_i O_j O_j)$ will equal either $3$ or $1$ depending on whether $O_i,O_j$ commute or anticommute.
Summing this trace over $j$, since both cases are roughly equally likely, we get roughly $2n$.
Hence, we have
\begin{eqnarray}
\langle O_i \rangle & \approx& \bigl(1+\frac{1}{2}\beta^2 n (1+J^2)+\ldots\bigr)^{-1} \bigl(\beta J+\frac{n-1}{2} \beta^3 J^3+\frac{1}{3!} \beta^3 J^3 +\frac{1}{3!} \beta^3 J (2n)+\ldots \bigr) \\ \nonumber
&=&
\beta J-\frac{1}{3} \beta^3 J^3-\frac{1}{6} \beta^3 J n+\ldots
\end{eqnarray}
The terms $\beta J-(1/3)\beta^3J^3$ are the first two terms in the expansion of $\tanh(\beta J)$. The additional term arises from anticommutation. More precisely, if $k_i$ is the number of $X$-type generators anticommuting with $O_i$, the trace sum above is $3n-2k_i$, and for fixed finite $n$,
\be
\langle O_i\rangle
=\beta J-\frac13\beta^3J^3-\frac13\beta^3Jk_i+O(\beta^5).
\ee
For a dense row with $k_i\simeq n/2$, this gives the displayed approximation. Since $\kappa=\beta^2d$, the anticommutation correction is $-(k_i/3d)\beta J\kappa$, which is of order $\beta J\kappa$ when $k_i=\Theta(d)$. The term $-(1/3)\beta^3J^3$ is smaller at large $n$ and fixed $\kappa$.

This series expansion gives the claimed decrease in ratio for both $X$-type and $Z$-type operators; we have done the calculation only for $Z$-type but we see a decrease for both $J<1$ and $J>1$.

We could proceed to higher orders in the series expansion to explore larger $\beta$, but instead we use numerical simulation.

\subsection{Numerical Simulation}
\label{subsec:ns}
With the negative signs in $H_A(J)$, the Hamiltonian has no sign problem in the $Z$ basis.
We performed a path integral Monte Carlo.  We used $50$ time slices (more slices may improve the simulation).  

We started with all physical spins on all time slices intially equal to $+1$, so that each $D_i$ initially equals $+1$ and each $Z$-type Hamiltonian contribution $-JD_i$ equals $-J$.
We included Monte Carlo moves that flip a spin at a single time slice, and also
we added a type of nonlocal move to the simulation.  For each $i$, this move flips spins in the $Z$ basis so that the sign of operator the $Z$-type operator $O_{i+n}$ flips in expectation, but all other $Z$-type operators do not have their expectation value change.  We do this flip at \emph{all} time slices so that also the expectation value of the $X$-type operators does not change.  Indeed, for simplicity, we do not explicitly do this move by flipping spins, but rather we consider a Monte Carlo move that leaves the spin unchanged but flips the sign of the coefficient in front of the $Z$-type operators.

This move helps equilibration; if this move is not done, then it is very hard for the system to equilibrate the $Z$-type operators.  Indeed, for the given initial configuration, a single spin flip would change the sign of roughly half the $Z$-type operators, costing a large amount of energy, so that the move is likely to be rejected even at fairly high temperatures.  With this nonlocal move, we can flip a single operator at a time, giving correct behavior at high temperature.

We report simulations below on a single random instance for $n=50,100,200$.  Similar behavior was shown for other random instances.  However, we show in some cases multiple runs for the same instance to show issues with equilibrating the Monte Carlo.
We did $5\times 10^6$ Monte Carlo steps before measuring observables, and then did $5\times 10^6$ Monte Carlo steps to measure observables, averaging data over those steps.

In \cref{fig1} we plot the ratio for the sum of $Z$-type operators for $J=0.9,1.1$ for $n=50,100,200$ as a function of inverse temperature $\beta$.  The $x$-axis is scaled by a factor proportional to $1/\sqrt{n}$, and we see a good collapse of the curves for $n=50,100,200$ on top of each other for given $J$.
However, we observe one surprising effect: the $J=0.9,n=200$ ratio follows the $J=0.9,n=50,100$ ratios at smaller $\beta$, but then at large $\beta$ it jumps up close to $1$.  We believe that this is a result of a failure to equilibrate the quantum Monte Carlo simulation: the initial configuration is well-polarized in the $Z$-direction and while equilibrium should have to a smaller expectation value of $Z$-type operators combined with a larger expectation value of $X$-type operators (which in this kind of quantum Monte Carlo simulation translates to a larger  \emph{entropy} at the cost of a higher \emph{energy}), the system can get stuck in the $Z$-polarized basin and be unable to explore the full space.

\begin{figure}
\includegraphics[width=0.5\textwidth]{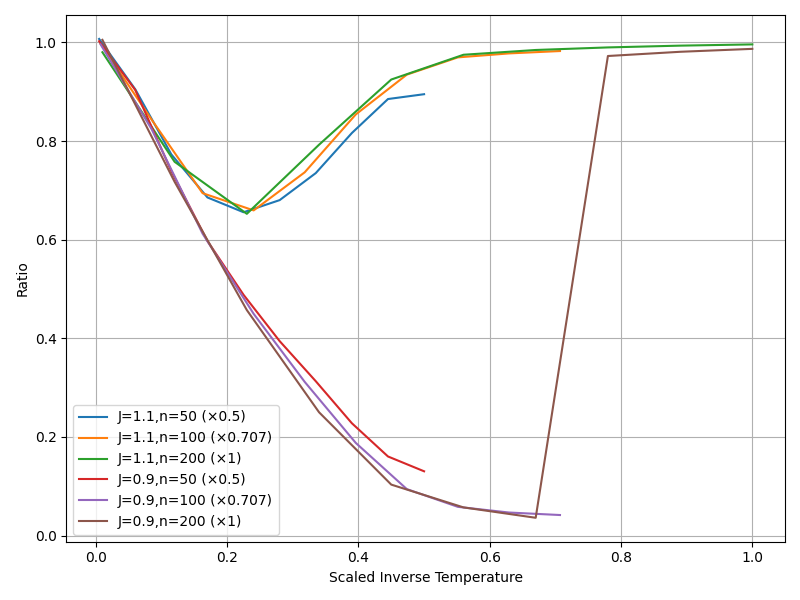}
\caption{Plot of ratio against inverse temperature for $J=0.9,1.1$, and $n=50,100,200$.  Inverse temperature is multiplied by a scale factor, given in parenthesis in the legend.  The scale factor is proportional to $1/\sqrt{n}$.}
\label{fig1}
\end{figure}

In \cref{fig2}, we plot multiple independent Monte Carlo runs for $J=0.9,n=200$.  We see that this jump in ratio appears in all of them, but it appears at slightly different inverse temperature in different ones.  In sum cases the ratio does not jump fully to $1$; this may be due to the system leaving the $Z$-polarized basin partway through the simulation, as observables are plotted averaged over the run after the initial warmup.

\begin{figure}
\includegraphics[width=0.5\textwidth]{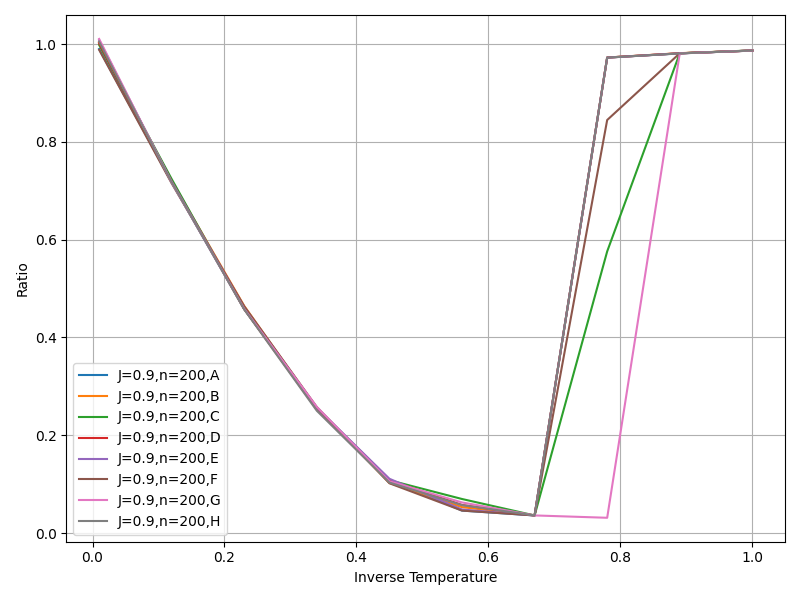}
\caption{Plot of ratio against inverse temperature for $J=0.9,n=200$.  Multiple Monte Carlo runs with the same random realization.}
\label{fig2}
\end{figure}

In \cref{fig3}, we show that failure to equilibrate can occur even for $J>1$ for this initial condition.  We report multiple independent Monte Carlo runs with the same random realization.  All but one of the runs follow a smooth curve, but one run has the ratio jump to a smaller value for one value of inverse temperature, so failure to equilibrate is much less frequent.  

\begin{figure}
\includegraphics[width=0.5\textwidth]{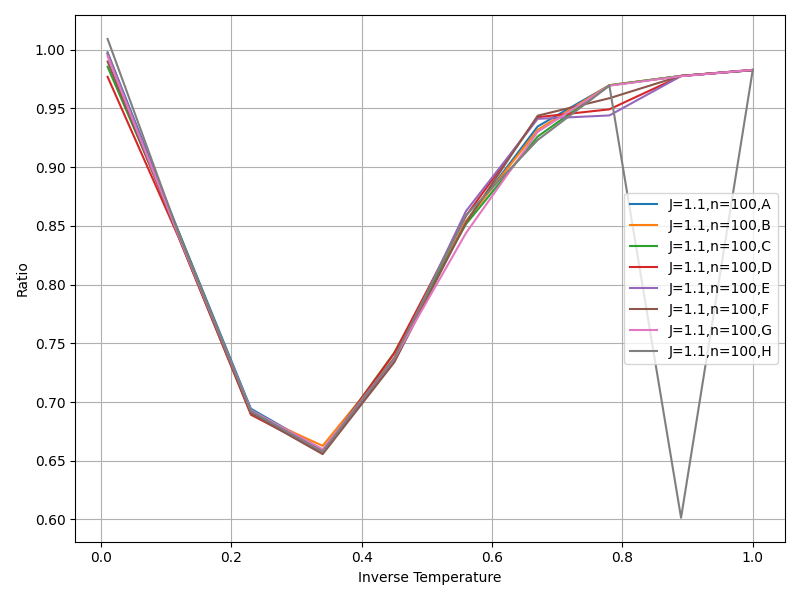}
\caption{Plot of ratio against inverse temperature for $J=1.1,n=100$.  Multiple Monte Carlo runs with the same random realization.}
\label{fig3}
\end{figure}

We leave it as an open problem to improve the equilibration of the Monte Carlo in this kind of model.  We did simulations independently at different temperature.  Using a simulation at a given temperature as a warmup for a nearby temperature may help.  Alternatively, one may try parallel tempering in both temperature and $J$.

\subsection{Toy Model}
\label{subsec:tm}
Finally, we consider the behavior in a toy analytic model. 
Let us first neglect the $X$-type terms and consider only $Z$-type terms.
Consider a typical thermal state in the $Z$-basis of this Hamiltonian with only $Z$-type terms, and call this state $|0\rangle$.  Let $E=\langle0|-J\sum_iD_i|0\rangle$ be the energy of the $Z$-type part in this state. Its thermal mean is $-nJ\tanh(\beta J)$, so typically $E\simeq-\beta J^2n$ when $\beta J\ll1$ and $\beta J\sqrt n\gg1$. Since $d=\Theta(n)$, its magnitude is then of order $J^2\sqrt\kappa\sqrt n$.

We now consider the $X$-type terms. For $i=1,\ldots,n$, let $|i\rangle=O_i |0\rangle$.  
Each of these states has some expectation value for the sum of $Z$-type operators, but since a single spin flip will typically flip about half of the $Z$-type terms, this expectation value can be approximated by just considering the expectation value in a random state, and hence the expectation value is typically proportional to $\sqrt{n}$.  Hence, for large enough $\kappa$, this expectation value is much smaller than $-E$ for most of these states $|i\rangle$.
Hence, considering just the states $|0\rangle,|i\rangle$, we have a Hamiltonian with $n+1$ basis states, a diagonal term on $|0\rangle$ that is proportional to $-\beta J^2 n$ and hopping to all the other basis states $|i\rangle$ with hopping proportional to $1$.  Ignoring the diagonal terms on those states $|i\rangle$, we can consider only the two basis states $|0\rangle$ and $(1/\sqrt{n}) \sum_i |i\rangle$, with hopping proportional to $\sqrt{n}$ between them.

Thus the collective hopping has magnitude $\sqrt n$, but is small compared with the diagonal splitting $|E|$: their ratio is of order $1/(\beta J^2\sqrt n)$, and hence of order $\kappa^{-1/2}$ for fixed $J>0$ and $d=\Theta(n)$.
\subsection{Formal proof}
\label{sec:NP_formal}

Place an anticommutation glass at each vertex of a cubic graph, using the interaction \cref{combineacgAF} with $J=1$. Choosing one of the two operator families at each vertex defines a cut. We will show that the log partition function records the size of a maximum cut, up to a bounded error per vertex. Keeping a commuting subset of the Hamiltonian gives the lower bound; the upper bound comes from the thermal series.

\subsubsection{Construction}

Let $G=(V,E)$ be a simple cubic graph, with $N=|V|$ vertices and $M=|E|=3N/2$ edges. At each vertex $v$, place a block of $n$ qubits with Hamiltonian
\be
D_{v,i}=\prod_{j=1}^n Z_{v,j}^{(A_v)_{ij}},
\qquad H_v=-\sum_{i=1}^n(X_{v,i}+D_{v,i}),
\label{eq:glass_block}
\ee
where the invertible binary matrix $A_v$ will be chosen below. Thus each block has two commuting families, which we label $X$ and $Z$; the label $Z$ refers to the operators $D_{v,i}$.

For each edge $e=uv$, choose an orientation to fix the order of $u,v$, introduce $n$ ancilla qubits $c_{e,i}$, and add
\be
V_e=-\widetilde J\sum_{i=1}^n
\left(X_{u,i}D_{v,i}X_{c_{e,i}}+D_{u,i}X_{v,i}Z_{c_{e,i}}\right),
\qquad
H_G=\sum_{v\in V}H_v+\sum_{e\in E}V_e.
\label{eq:glass_coupling}
\ee
The two terms in parentheses give the two interaction types $(X,Z)$ and $(Z,X)$. If the endpoint blocks choose different families, one term agrees with both choices. If they choose the same family, neither term does.

For the matrices, take $n=4^s\geq16$ and label the qubits and generators within each block by $\mathbb F=\mathbb F_n$. Properly color $G$ with four colors, which is possible by a greedy coloring since the maximum degree is three. Assign a distinct nonzero field element to each color, and let $\alpha_v$ be the element assigned to the color of $v$. Set
\be
A_v(x,y)=\delta_{xy}+\operatorname{tr}_{\mathbb F}(\alpha_vxy),
\label{eq:glass_matrices}
\ee
where $\operatorname{tr}_{\mathbb F}$ is the field trace to $\mathbb F_2$ and addition is over $\mathbb F_2$. These matrices are symmetric and satisfy $A_v^2=I$. Every row with nonzero label has weight $n/2\pm1$. The four choices give blocks related by permutations of their qubits and generators. The sign estimates are proved in Appendix~\ref{app:glass_finite}.

The ancillas ensure that adding the couplings has not introduced relations. Indeed, in any product proportional to the identity, the $X$ and $Z$ exponents on each ancilla must both be even. Hence every coupling term occurs an even number of times. The remaining onsite generators are independent because each $A_v$ is invertible. Thus $H_G$ is redundancy-free.

An onsite term has at most $4n$ anticommuting neighbors. An edge term has at most $2n$ such neighbors among onsite terms, $4n$ on other incident edges, and $n$ in the other term of its own edge. On the other hand, a nonzero-labeled row of any $A_v$ already gives at least $n/2-1$ anticommuting neighbors. Therefore the maximum anticommutation degree satisfies
\be
\frac n2-1\leq d\leq7n.
\label{eq:glass_degree}
\ee

Fix positive integers $t,b$ with $b\geq t$, and set
\be
\beta=\frac b{\sqrt n},
\qquad \widetilde J=\frac tb,
\qquad K=\frac{t^2}{2}.
\label{eq:glass_scales}
\ee
All coefficients then have magnitude at most one, and $\beta=\Theta(d^{-1/2})$ for fixed $b$. The quantity $K$ will be the gain in the log partition function from one cut edge. We use the normalized partition function
\be
\mathcal Z_{n,G}=2^{-n(N+M)}\operatorname{Tr}e^{-\beta H_G}.
\label{eq:glass_normalization}
\ee
\begin{theorem}
\label{thm:glass_hardness}
There are positive integers $t\leq b$ and $n_0$ with the following property. Fix any power of four $n\geq n_0$. Then, given a cubic graph $G$ on $N$ vertices, approximating $\mathcal Z_{n,G}(\beta)$ within a factor of two is NP-hard at inverse temperature $\beta=\Theta(d^{-1/2})$, with constants independent of $G$. The corresponding Hamiltonians $H_G$ are redundancy-free, have coefficients of magnitude at most one, and satisfy \cref{eq:glass_degree}. The block size $n$ is independent of the graph size $N$.
\end{theorem}

\subsubsection{A commuting lower bound}

First choose a family $\sigma_v\in\{X,Z\}$ at each vertex. Retain only that onsite family. On every edge whose endpoints choose different families, also retain the one coupling term agreeing with both choices. Remove all other terms.

Every retained term commutes with every other. At a block, all retained operators belong to the same commuting family, and each ancilla occurs in at most one retained term. The retained generators are also independent, since they form a subset of the original generators. Consequently their normalized partition function is a product of coshes:
\be
\mathcal Z_\sigma=(\cosh\beta)^{nN}
(\cosh(\beta\widetilde J))^{n\operatorname{cut}_G(\sigma)},
\label{eq:glass_commuting_partition}
\ee
where $\operatorname{cut}_G(\sigma)$ counts edges with different endpoint labels. To see the factorization, expand each exponential as $e^{hP}=\cosh h+P\sinh h$. Independence makes the trace of every nonempty product of distinct generators vanish, leaving only the product of coshes.

Removing a term from a redundancy-free Hamiltonian cannot increase its partition function. Indeed, fix all other couplings and write $g(x)=\log\operatorname{Tr}e^{Q+xP}$, where $P$ is the term being removed. Every word with an odd number of appearances of $P$ has zero trace, so $g$ is even. It is also convex, and hence $g(x)\geq g(0)$. Applying this one term at a time, and choosing a maximum cut, gives
\be
\log\mathcal Z_{n,G}\geq
nN\log\cosh\left(\frac b{\sqrt n}\right)
+n\log\cosh\left(\frac t{\sqrt n}\right)\operatorname{MAXCUT}(G).
\label{eq:glass_lower_exact}
\ee
To bound the finite-$n$ error, use
\be
0\leq\frac{x^2}{2}-\log\cosh x\leq\frac{x^4}{12}.
\label{eq:glass_cosh_bound}
\ee
For $x\geq0$, this follows by integrating $x-x^3/3\leq\tanh x\leq x$; the general case follows by evenness. Since $\operatorname{MAXCUT}(G)\leq3N/2$, we obtain
\be
\log\mathcal Z_{n,G}\geq\frac{Nb^2}{2}+K\operatorname{MAXCUT}(G)
-N\left(\frac{b^4}{12n}+\frac{t^4}{8n}\right).
\label{eq:glass_lower}
\ee
\subsubsection{A positive chord model}

Expand the trace in the individual Pauli terms, now keeping words from both families. Redundancy-freeness removes every word in which some generator appears an odd number of times. In each remaining word, pair the appearances of each generator, place the positions of the word around a circle, and join paired positions by chords.

The value of the word is determined by crossings. For example, $PQPQ$ equals $I$ or $-I$ according to whether $P$ and $Q$ commute or anticommute; its two chords cross. The word $PPQQ$ gives $I$, and its chords do not cross. More generally, bringing the two operators in each pair together gives one commutation sign for every crossing of chords. Summing over the microscopic generator indices will largely cancel diagrams with crossings between opposite families. We first study the positive model in which these crossings are forbidden, and then control the finite-size error in this replacement.

In the comparison model, a chord joins two independent uniform points on a circle of circumference one. An onsite chord records its vertex and its family, $X$ or $Z$. An edge chord records its edge and its interaction type, $(X,Z)$ or $(Z,X)$. It uses the same two circle coordinates at both endpoint blocks, since the two factors belong to a single Hamiltonian term. We call each choice of vertex and family, or of edge and interaction type, a species.

Each onsite species has intensity $b^2/2$, and each edge species has intensity $K$. These are the second-order weights after summing the microscopic index: there are $n$ choices, giving $n\beta^2/2=b^2/2$ onsite and $n(\beta\widetilde J)^2/2=t^2/2=K$ on an edge. More precisely, a species of intensity $a$ assigns the unnormalized Poisson weight
\be
\frac{a^k}{k!}\prod_{j=1}^k ds_j\,dt_j
\label{eq:glass_chord_measure}
\ee
to $k$ chords with endpoints $s_j,t_j$. Its total mass is $e^a$. Let $\nu$ be the product measure over all species. For a diagram $\mathcal D$, let $I_v(\mathcal D)$ be one if no two crossing chords have opposite family labels at $v$, and zero otherwise. Define
\be
F_G(b,K)=\int\prod_{v\in V}I_v(\mathcal D)\,d\nu(\mathcal D).
\label{eq:glass_positive_model}
\ee
Consider first one isolated block. Form the chord-intersection graph, with one vertex per chord and an edge between chords that cross. Compatibility requires every connected component to have a single family label. We use the label of a largest component as the phase $\sigma$ of the block. If there is a tie, choose a component using the cyclic endpoint order rooted at a uniformly selected endpoint, independently of the family labels and endpoint spacings. Give the empty diagram a fair auxiliary phase. This convention preserves the symmetry between $X$ and $Z$.

Let $F_b$ be the isolated-block partition function. Each phase has weight $F_b/2$; write $\mu_b^\sigma$ for the corresponding normalized conditional law. For an onsite diagram of phase $\sigma$, let $\phi$ be the probability that a fresh uniform chord crosses none of the onsite chords with family $\sigma$. A coupling chord using the opposite family can be compatible only on this part of its endpoint space. Thus a bound on $\phi$ will bound the gain from a interaction that disagrees with the phase of a block.

\begin{lemma}
\label{lem:glass_onsite}
There is an absolute constant $B_0$ and a  function $\delta_b\to0$ as $b\to\infty$ such that
\be
F_b\leq B_0e^{b^2/2},
\qquad \mathbb E_{\mu_b^\sigma}\phi\leq\delta_b.
\label{eq:glass_onsite_bounds}
\ee
One may take $B_0=e^{16}$. The second estimate holds for both phase choices.
\end{lemma}

At large $b$, a typical onsite diagram has a crossing-connected component containing most of its chords. A chord avoiding this component must put both endpoints in the same gap between consecutive endpoints of the component, and these gaps are small. Appendix~\ref{app:glass_onsite} proves the lemma using the large-component theorem for uniform chord diagrams~\cite{acan2017chord}.

\subsubsection{The upper bound}

Fix all onsite diagrams and their phases, and drop compatibility conditions between edge chords. This can only increase the integral. The edge chords then interact only with the fixed onsite diagrams, so we can sum them independently on each edge.

Suppose first that $\sigma_u\neq\sigma_v$. One interaction term agrees with the phases at both endpoints. Its allowed endpoint area is at most one. The other interaction term disagrees at both endpoints, so its allowed area is at most $\min(\phi_u,\phi_v)$, and in particular at most $\phi_u+\phi_v$. If instead $\sigma_u=\sigma_v$, each term disagrees at one endpoint, and the sum of their allowed areas is again at most $\phi_u+\phi_v$.

A species of intensity $K$ restricted to an endpoint region of area $a$ has total weight $\sum_{k\geq0}(Ka)^k/k!=e^{Ka}$. Hence the two interaction types on an edge together contribute at most
\be
\exp\left[K\mathbf1\{\sigma_u\neq\sigma_v\}
+K(\phi_u+\phi_v)\right].
\ee
Let $W_G(\sigma)$ be the contribution to $F_G$ with all phases fixed to $\sigma$. Multiplying the edge bounds and integrating the independent onsite laws gives
\be
W_G(\sigma)\leq(F_b/2)^N e^{K\operatorname{cut}_G(\sigma)}
\prod_{v\in V}\mathbb E_{\mu_b^{\sigma_v}}
 e^{K\deg(v)\phi_v}.
\label{eq:glass_phase_upper}
\ee
Each $\phi_v$ appears once for every edge incident to $v$. Since $0\leq\phi_v\leq1$ and $\deg(v)=3$, convexity gives
\be
\mathbb E e^{3K\phi_v}
\leq1+(e^{3K}-1)\delta_b
\leq e^{\varepsilon_b},
\qquad \varepsilon_b=(e^{3K}-1)\delta_b.
\label{eq:glass_screening_error}
\ee
There are $2^N$ phase assignments, and every cut has size at most $\operatorname{MAXCUT}(G)$. Summing \cref{eq:glass_phase_upper} therefore gives
\[
F_G(b,K)\leq F_b^N
 e^{K\operatorname{MAXCUT}(G)+N\varepsilon_b}.
\]
The factor $2^N$ from the sum cancels the factor $2^{-N}$ in the fixed-phase onsite weights. Applying \cref{lem:glass_onsite} now bounds $\log F_G$ by $Nb^2/2+K\operatorname{MAXCUT}(G)+N(\log B_0+\varepsilon_b)$.

The finite-size comparison is given by the following lemma.

\begin{lemma}
\label{lem:glass_finite_upper}
For the matrices in \cref{eq:glass_matrices} and fixed $b,K$, there is a constant $C_{b,K}$, independent of $G,N,n$, such that
\be
\mathcal Z_{n,G}\leq F_G(b,K)
\left(1+C_{b,K}n^{-1/10}\right)^N
\label{eq:glass_finite_upper}
\ee
for every simple graph of maximum degree at most three and every power of four $n\geq16$.
\end{lemma}

Appendix~\ref{app:glass_finite} proves the lemma by averaging the microscopic labels at incompatible crossings and removing constraints near those crossings. Combining it with the bound on $F_G$ gives
\be
\begin{split}
\log\mathcal Z_{n,G}\leq{}&\frac{Nb^2}{2}+K\operatorname{MAXCUT}(G)\\
&+N\left[\log B_0+\varepsilon_b+
\log(1+C_{b,K}n^{-1/10})\right].
\end{split}
\label{eq:glass_upper}
\ee
\subsubsection{The reduction}

There are fixed constants $a_+>a_-$ for which distinguishing
\be
\operatorname{MAXCUT}(G)\geq a_+M
\qquad\hbox{from}\qquad
\operatorname{MAXCUT}(G)\leq a_-M
\label{eq:glass_gap_promise}
\ee
on cubic graphs is NP-hard~\cite{berman2001amplifiers}. Put $\gamma=a_+-a_->0$. The gap in the desired contribution to $\log\mathcal Z_{n,G}$ is $K\gamma M=(3/2)K\gamma N$. We first make this large by choosing an integer $t$ such that, with $K=t^2/2$,
\be
\frac32K\gamma>\log B_0+5.
\label{eq:glass_choose_K}
\ee

Increasing $K$ also increases the factor multiplying $\delta_b$ in \cref{eq:glass_screening_error}. We therefore choose $b$ after fixing $K$: take an integer $b\geq t$ for which $\varepsilon_b<1$. Finally choose $n_0\geq16$ large enough that, for all $n\geq n_0$,
\be
\frac{b^4}{12n}+\frac{t^4}{8n}\leq2,
\qquad \log(1+C_{b,K}n^{-1/10})\leq1.
\label{eq:glass_choose_n}
\ee
These constants are independent of the graph size $N$.

For any power of four $n\geq n_0$, \cref{eq:glass_lower,eq:glass_upper} give
\be
\begin{split}
\frac{Nb^2}{2}+K\operatorname{MAXCUT}(G)-2N
&\leq\log\mathcal Z_{n,G},\\
\log\mathcal Z_{n,G}
&\leq\frac{Nb^2}{2}+K\operatorname{MAXCUT}(G)
 +(\log B_0+2)N.
\end{split}
\label{eq:glass_sandwich}
\ee
Thus the lower bound for a YES instance exceeds the upper bound for a NO instance by
\be
K\gamma M-(\log B_0+4)N>N.
\ee
A factor-two approximation changes the logarithm by at most $\log2$. It therefore distinguishes the two cases for all sufficiently large $N$; the finitely many smaller cases can be handled directly.

There are $n(N+M)$ physical qubits and $2n(N+M)$ terms. For fixed admissible $n$, both counts are linear in $N$, and the construction takes polynomial time. This proves \cref{thm:glass_hardness}.

\section{Correlations and entanglement}
\Cref{sec:polymer,sec:QSP} show that redundancy-free Hamiltonians
are algorithmically tractable whenever $\kappa=\beta^2 d$ is small.
This regime extends to inverse temperatures $\beta\sim 1/\sqrt d$, quadratically
beyond the scale $\beta\sim 1/d$ at which generic Hamiltonians of the same
degree remain tractable.

In this section we quantify the strength of  correlations and entanglement in redundancy-free Hamiltonians. These Hamiltonians exhibit genuine quantum properties, even in the regime where they are tractable: just like general Hamiltonians, we show that redundancy-free Hamiltonians can be entangled already at inverse temperature $\beta = O(1/d).$ In \cref{sec:QMA_hard}, we further prove that their ground state problem is $\mathsf{QMA}$-complete. 

Interestingly, we show that correlations in redundancy-free Hamiltonians are \emph{more} localized than for generic Hamiltonians:  Specifically, their decay is quadratically stronger in their dependence on the inverse-temperature $\beta$. We also show a similar quadratically stronger bound on the mutual information between different regions. 

For comparison with results on general
Hamiltonians we will sometimes use the \emph{interaction degree} $d_{\mathrm{int}}$, the maximum
number of terms whose support intersects the support of a given term. Since two
Pauli operators with disjoint supports commute, $d\leq d_{int}$.

\subsection{Decay of Correlations}
\label{corrdecay2}

General Hamiltonians with bounded interaction-degree $d_{\rm int}$ are known to exhibit exponential decay of correlations for sufficiently small $\beta d_{\rm int}$, due to Ref.~\onlinecite{Kliesch2014Locality}. For a fixed local observable $O_A$ and a distant observable $O_C$, the usual high-temperature bounds take the form
\be
\left|
\langle O_A O_C\rangle_{\rho_\beta}
-
\langle O_A\rangle_{\rho_\beta}
\langle O_C\rangle_{\rho_\beta}
\right|
\leq
K_A \|O_C\| e^{-\ell/\xi_{\rm gen}},
\ee
where $\ell$ measures the separation in interaction steps and $K_A$ is a local prefactor. At sufficiently high temperatures, the decay factor $e^{-1/\xi_{\rm gen}}$ can be bounded by $O(\beta d_{\rm int})$. The reason for this scale is that a connected sequence of interactions joining the two observables contributes a factor of order $\beta$ for each interaction, while the number of possible sequences grows with the interaction-degree.

For redundancy-free Hamiltonians in the canonical realization \eqref{Oirep}, we instead obtain decay controlled by $\kappa=\beta^2d$, where $d$ is the degree of the anticommutation graph. Thus, exponential decay persists up to inverse temperatures $\beta=O(1/\sqrt d)$, even though, as shown in the previous subsection, the thermal state may already be entangled at $\beta=O(1/d)$. The improvement comes from the same even-occurrence property that controls the partition function: every generator contributing to a correlation through the intervening region must occur at least twice. Its contribution is therefore of order $\beta^2$ rather than $\beta$.

We will prove a version of this statement for the thermofield double, rather than only for the thermal state. In particular, the distant observable may act on both the left and right copies, and we impose no restriction on the size of its support. This distinction is useful for state preparation, where we need to control the correlation of a local Hamiltonian term with an entire distant region, not just with another local observable. Restricting the distant observable to the left copy gives the corresponding statement for the thermal state.

To state the result, fix a term $O_i$ and let $A$ be its support in the canonical realization. Let $C$ be a distant set of sites and let $B$ contain the remaining sites, separating $A$ from $C$. Distances are measured in the anticommutation graph, whose vertices label the physical qubits in this realization. We say that the buffer has width $\ell$ if every path from $A$ to $C$ contains at least $\ell$ intermediate vertices in $B$. Each region includes its corresponding sites in both copies of the thermofield double.

\begin{lemma}[Decay of correlations]
\label{lem:rfree_correlation_decay}
Let $H=\sum_j J_j O_j$ be a redundancy-free Hamiltonian in the canonical realization \eqref{Oirep}, with $|J_j|\leq 1$ and anticommutation-graph degree $d\geq 1$. There are constants $C_0,C_1,\kappa_0>0$ such that, whenever $\kappa=\beta^2d\leq\kappa_0$, with $C_1\kappa_0<1$, we have
\be
\label{eq:rfree_correlation_decay}
\left|
\left\langle
\left(O_{i,L}-\langle O_{i,L}\rangle\right)O_C
\right\rangle
\right|
\leq
C_0(C_1\kappa)^\ell \|O_C\| .
\ee
Here the expectation values are in the normalized thermofield double, and $O_C$ is any operator supported on the left and right copies of $C$. The constants are independent of the system size, the buffer width, and the size of the support of $O_C$.
\end{lemma}

One may first think that the polymer expansion could be used to prove this as we might add a term proportional to $O_C$ to the Hamiltonian.  However, since we have no bound on the support of $O_C$, this method would perhaps run into issues, so we will use a different method.

We take two copies of the thermofield double, calling those copies $1,2$.  So, there are a total of four Hilbert spaces: $(1,L),(1,R),(2,L),(2,R)$.
Two copies of the thermofield double state is the same as
the thermofield double state of the Hamiltonian $H_1+H_2$.
We have
\be
\label{twocopy}
\langle (O_{i,L} - \langle O_{i,L} \rangle ) O_C \rangle=\langle (O_{i,(1,L)}-O_{i,(2,L)}) O_{C,1}\rangle,
\ee
where $O_{C,1}$ denotes $O_C$ acting on copy $1$.

Expand the right-hand side of \cref{twocopy} in a series in $H$ by writing 
the two copies of the thermofield double state
$Z^{-2} \exp[-\beta (H_{1,L}+H_{2,L})/2] \Psi_{ME,1}\otimes \Psi_{ME,2}$,
and expanding 
$\exp[-\beta (H_{1,L}+H_{2,L}]$ as a series in $H$ (do not expand $Z$), and then further expanding that series in $H$ by writing powers of $H$ as sums of products of operators.  Here the operators in the product will be $O_{i,(1,L)}+O_{i,(2,L)}$, i.e., the sum of the given operator on both copy $1$ and copy $2$.
The expectation value vanishes unless every operator supported in $B$ appears an even number of times in this expansion, so we may drop terms where it appears an odd number of times.

For each term $T$, consider the set containing all $i$ which are operators in that term $T$ as well as all sites in $A$.  Let $a(T)$ be the connected component containing $A$.

Note that any term $T$ for which $a(T)$ does not include any sites in $C$  gives a vanishing contribution to the expectation value of 
$(O_{i,(1,L)}-O_{i,(2,L)}) O_{C,1}$ because the contribution to the expectation of
$O_{i,(1,L)} O_{C,1}$ cancels the contribution to the expectation of $O_{i,(2,L)} O_{C,1}$.  This is why we introduced these two copies, $1,2$, to make cancellation between the product of the expectation value and the expectation value of the product more explicit.

We now bound the sum, over terms $T$ such that $a(T)$ does include some site in $C$, of the contribution to the expectation value.
Interchange the order of summation, so we sum over animals $a$ containing $A$ and also some site in $C$, and then sum over all terms $T$ such that $a(T)=a$, i.e., such every site in $a\setminus A$ appears in the term, and so that no sites neighboring $a$ appear in the term.

Consider given $a$.  Let $H'$ be the sum of terms in the Hamiltonian not including any site in $a$ or neighboring $a$.
Let $Z'$ be the partition function for $H'$.
In this paragraph, when we say norm, we refer to the $\ell_2$ norm of a state given by acting with some sum of terms on $\Psi_{ME}$.
The norm of the sum of terms including no sites in $a$ or neighboring $a$ is equal to $(Z'/Z)^2$, which for redundancy-free Hamiltonians is $\leq 1$.
So, the norm of the sum of terms in which each site $i$ in $a$ appears some number of times $n_i$ is bounded by by $$(Z'/Z)^2 \prod_{i\in a} (1/n_i!) (\beta J_i)^{n_i} \leq \prod_{i\in a} (\beta^{n_i}/n_i!).$$
Restricting to terms with $n_i$ even and positive, we can sum over $n_i$ and for $\beta<1$, the sum is bounded by
$O(\beta^2)^{|a|}$.

 We must have $|a|\geq \ell$.
Since the number of lattice animals of cardinality $|a|$ is bounded by $O(d)^{|a|}$, for sufficiently small $\kappa$
the sum of $O(\beta^2)^{|a|}$ over all $a$ is bounded by $O(\kappa)^{\ell}$.

Remark: a very similar method can show decay of correlations for general Hamiltonians for sufficiently small $\beta d$ on an interaction graph of degree $d$ as defined in the next section; we now only have that $Z'\leq Z \exp(O(\beta) d|a|)$ as removing terms from $H$ to define $H'$ may increase the partition function in this case, but this is sufficient.
\subsection{Separability}
General Hamiltonians with bounded interaction-degree $d$ are known to be separable below an inverse temperature $\beta = O(d_{int}^{-1})$ due to REF. Likewise, there exist generic methods for preparing the TFD on at least lattices (details) up to a temperature scale  $\beta = O(1/d_{int})$ in polynomial-time, whereas our quantum state preparation protocol achieves $\beta = O(1/\sqrt d)$. 
One may therefore ask if this is the case simply because redundancy-free Hamiltonians are separable up to $\beta =  O(1/\sqrt d)$ and only start to show entanglement beyond this point. 

We first give the following simple general construction shows that this is not true. 
Start with any Hamiltonian that is entangled already at some $\beta =  O(1/ d)$. For concreteness, we use a simple bipartite antiferromagnetic Heisenberg Hamiltonian. Take two sets $A = {a_1, ... , a_r}$ and $B = {b_1, ... , b_r}$ of $r$ qubits and define the Hamiltonian \begin{equation}H_r = \sum_{a \in A, b \in B} ( X_a X_b + Y_a Y_b + Z_a Z_b). \end{equation} Its  partition function on the sector 
$S_A=S_B=r/2$ is \begin{equation}
Q_r(\beta)=\sum_{\ell=0}(2 \ell+1) e^{-2 \beta \ell(\ell+1)} .
\end{equation}
The unique spin-zero state is maximally entangled between two $(r+1)$-dimensional spaces, and its thermal population is $1/Q_r(\beta)$. A PPT state can have overlap at most $1/(r+1)$ with this singlet, so
\begin{equation}
Q_r(\beta)<r+1 \quad \Longrightarrow \quad \rho_\beta\left(H_r\right) \text { is NPT across } A \mid B .
\end{equation} Hence this Hamiltonian is entangled already at inverse temperature $\beta = 1/2r = O( 1/d )$.

The above construction is 2-local but not redundancy-free. To remove redundancies, we simply multiply each term with a Pauli $Z$ operator on a fresh ancilla. The resulting redundancy-free Hamiltonian is \begin{equation}
\widetilde{H}_r=\sum_{a, b}\left(X_a X_b Z_{c_{a b, x}}+Y_a Y_b Z_{c_{a b, y}}+Z_a Z_b Z_{c_{a b, z}}\right) .
\end{equation} If the original Hamiltonian $H_r$ fails a PPT test, so will the new Hamiltonian $\widetilde{H}_r$, as its thermal state is an incoherent mixture of different states, depending on the sign of the ancillas $Z$. Taking the partial transpose preserves this block decomposition, so its negative eigenvalue survives in the full state. Hence we have constructed a $3$-local, redundancy-free Hamiltonian that is not separable already at temperature $O(1/d)$. The above construction does not exaclty agree with the . One example is the canonical realization \eqref{Oirep} of the complete bipartite anticommutation graph $K_{d,d}$ with a left copy $L$ and a right copy $R$, ordered with all $L$ vertices before all $R$ vertices 
\begin{equation}
O_i=X_i \quad(i \in L), \quad O_j=X_j \prod_{i \in L} Z_i \quad(j \in R) .
\end{equation} We mention the ordering since the canonical realization depends on it. For simplicity, take $d$ even and all coupling strengths equal to one. The operators $X_j$ with $j\in R$ commute with the Hamiltonian. Restricting to the sector where all these operators have eigenvalue $+1$, and applying a Hadamard to each qubit in $L$, gives the conditional Hamiltonian
\begin{equation}
G_d=\sum_{i=1}^d Z_i+dX_1\cdots X_d .
\end{equation}
Each computational-basis string and its bitwise complement span an invariant subspace. On the subspace spanned by $|0^d\rangle$ and $|1^d\rangle$, the Hamiltonian is $d(Z+X)$, so the corresponding off-diagonal entry of $e^{-\beta G_d}$ is $-\sinh(\sqrt{2}\beta d)/\sqrt{2}$. For strings with equal numbers of zeros and ones, the diagonal entries are instead $\cosh(\beta d)$.

Split $L=L_1\cup L_2$ into two sets of $d/2$ qubits. Partial transposition on $L_1$ moves the off-diagonal entry above to the pair of strings $|1^{d/2}0^{d/2}\rangle$ and $|0^{d/2}1^{d/2}\rangle$. Hence the partially transposed, unnormalized Gibbs operator contains the block
\begin{equation}
\begin{pmatrix}
\cosh(\beta d) & -\sinh(\sqrt{2}\beta d)/\sqrt{2}\\
-\sinh(\sqrt{2}\beta d)/\sqrt{2} & \cosh(\beta d)
\end{pmatrix}.
\end{equation}
Its smaller eigenvalue is
\begin{equation}
\cosh(\beta d)-\frac{1}{\sqrt{2}}\sinh(\sqrt{2}\beta d),
\end{equation}
which is negative, for example, at $\beta=2/d$. Since the conditioned sector has positive thermal weight and partial transposition preserves the sector decomposition, the full thermal state is NPT across $L_1\mid(L_2\cup R)$. Thus this canonical redundancy-free Hamiltonian is also entangled already at inverse temperature $O(1/d)$.

\subsection{Mutual information}
General Hamiltonians are known to satisfy a thermal mutual information area law due to Wolf et al.~\onlinecite{Wolf2008AreaLaws}. For a Hamiltonian $H=\sum_i J_i O_i$ with $\|O_i\|\leq 1$, take a bipartition $A|B$ of the physical sites and let $\partial(A:B)$ denote the terms whose supports intersect both sides. Writing $H=H_A+H_B+H_\partial$, where $H_\partial$ contains precisely these crossing terms, the general bound is
\begin{equation}
I(A:B)_{\rho_\beta}
\leq 2\beta\|H_\partial\|
\leq 2\beta\sum_{i\in\partial(A:B)}|J_i| .
\end{equation}
Here $\rho_\beta=e^{-\beta H}/\operatorname{Tr}e^{-\beta H}$ and $I(A:B)=S(\rho_A)+S(\rho_B)-S(\rho_\beta)$. The proof compares the free energy of the thermal state with that of the product of its marginals. The terms supported entirely within $A$ or $B$ have the same expectation in both states, so only the interactions crossing the boundary contribute.

For redundancy-free Hamiltonians, we can strengthen this bound by an additional factor of $\beta$ at high temperatures. More precisely, we have
\begin{equation}
\label{eq:rfree_mutual_information}
I(A:B)_{\rho_\beta}
\leq
\sum_{i\in\partial(A:B)}
\beta |J_i|\tanh(\beta |J_i|)
\leq
\beta^2\sum_{i\in\partial(A:B)}J_i^2 .
\end{equation}
For unit-strength interactions, the bound is therefore $\beta^2|\partial(A:B)|$ rather than $2\beta|\partial(A:B)|$. The important point is that the squared couplings add term by term, with no dependence on the strengths or number of the remaining interactions. The bound holds at every finite $\beta\geq 0$, without a restriction involving the interaction-degree.

The conceptual reason is that removing a redundancy-free interaction makes its expectation exactly zero, even in the presence of all the other interactions. Restoring that interaction then produces an expectation proportional to its own coupling at high temperatures. This supplies the additional factor of $\beta$.

To make this precise, first remove all interactions crossing the boundary and consider the resulting thermal state
\begin{equation}
\sigma
=
\frac{e^{-\beta(H_A+H_B)}}{\operatorname{Tr}e^{-\beta(H_A+H_B)}}
=
\sigma_A\otimes\sigma_B .
\end{equation}
For every omitted term, $\operatorname{Tr}(\sigma O_i)=0$. Indeed, expanding $O_i e^{-\beta(H_A+H_B)}$ gives only words in which the omitted generator occurs exactly once, so every trace vanishes by redundancy-freeness. Since $\sigma$ is a product state, its relative entropy from $\rho_\beta$ bounds the mutual information. Using the vanishing expectations gives
\begin{equation}
\begin{aligned}
I(A:B)_{\rho_\beta}
&\leq D(\rho_\beta\|\sigma)\\
&\leq D(\rho_\beta\|\sigma)+D(\sigma\|\rho_\beta)\\
&=-\beta\operatorname{Tr}(\rho_\beta H_\partial).
\end{aligned}
\end{equation}
Here $D(\rho\|\sigma)=\operatorname{Tr}[\rho(\log\rho-\log\sigma)]$. Thus it remains to bound the expectation of each crossing interaction in the full thermal state.

For this, fix all other couplings and write $g(x)=\log\operatorname{Tr}e^{K-xO}$, where $K$ contains only the remaining generators and $\|O\|\leq 1$. Redundancy-freeness makes $g$ even, since every term with an odd number of occurrences of $O$ has zero trace. Hence $g'(0)=0$. Moreover, $g'(x)=-\langle O\rangle_x$, and its derivative is a thermal susceptibility bounded by the variance of $O$:
\begin{equation}
0\leq g''(x)
\leq \langle O^2\rangle_x-\langle O\rangle_x^2
\leq 1-g'(x)^2 .
\end{equation}
Integrating from $x=0$ gives $0\leq g'(x)\leq\tanh x$ for $x\geq 0$. Taking $O=\operatorname{sgn}(J_i)O_i$ and $x=\beta|J_i|$ therefore yields
\begin{equation}
0\leq
-J_i\operatorname{Tr}(\rho_\beta O_i)
\leq |J_i|\tanh(\beta|J_i|).
\end{equation}
Summing this estimate over the crossing interactions proves \eqref{eq:rfree_mutual_information}.

We emphasize that this argument applies to any redundancy-free realization, not only the canonical realization \eqref{Oirep}. The boundary counts terms crossing the physical bipartition, rather than edges of the anticommutation graph. The stronger bound also does not imply separability: the thermal state may be entangled while its mutual information across the same cut is small, as in the examples above.

\section{Quantum State Preparation}
\label{sec:QSP}
We now consider preparing thermal states of these Hamiltonians on a quantum computer for $\beta \lesssim 1/\sqrt{d}$.   In fact, we will show how to prepare thermofield double states, which allows one to prepare a thermal state by tracing out the ancilla.

Before showing how to do this, we remark that it is 
trivially possible to efficiently prepare a quantum state for which the expectation value of every term $O_i$ is $\Theta(1/\sqrt{d})$, i.e., a state in which all expectation values have the same scaling with $d$ as they would in a thermal state at $\beta \sim 1/\sqrt{d}$.
Recall the representation of the operators $O_i$ on $n$ qubits as 
$O_i=X_i \prod_{j<i, (i,j)\in E} Z_j$.
Take a pure product state where each qubit lies in the $X-Z$ plane with $\langle Z_i \rangle=\cos(\theta_i)$ and $\langle X_i \rangle=\sin(\theta_i)$.
Take all the $\theta_i$ all equal to some fixed angle $\theta$.

Then, $\langle O_i \rangle = \sin(\theta) \cos(\theta)^d_i$, where $d_i$ is the degree of $i$ in the anticommutation  graph.  Thus,
$\langle O_i \rangle \geq \theta (1-d \theta^2/2)+O(\theta^5)$.  For large $d$, we have small $\theta$, and this is optimized at $\theta=\sqrt{3/2} \cdot \sqrt{1/d}$,
giving $\langle O_i \rangle \geq (1/4) \cdot \sqrt{3/2} \cdot \sqrt{1/d}$.

Of course, if the degree $d_i$ depends on $i$, one can perhaps improve the expectation values by taking $\theta_i$ depending on $i$.  We ignore this and now consider the question of preparing the thermal state.

We will give a method for preparing the thermofield state approximately assuming a certain correlation decay;
this method does not depend on the Hamiltonian being redundancy-free and only depends on locality.
From \cref{lem:rfree_correlation_decay}, this correlation decays holds for redundancy-free Hamiltonians for $\beta\lesssim 1/\sqrt{d}$.
The runtime of resulting algorithm is superpolynomial, but subexponential in $n/\epsilon$, where $\epsilon$ is the error.  (Remark: on an anticommutation graph where the number of vertices within a given distance of a given vertex grows only polynomially, the resulting algorithm becomes quasipolynomial.)
Later, in \cref{sec:beyond}, we propose a possible route to polynomial time preparation algorithms, and give an alternative algorithm for general Hamiltonians.

Consider a general Hamiltonian $H=\sum_i J_i O_i$, where $\Vert O_i \Vert \leq 1$, and
where we assume that there is some graph of degree at most $d$, with a qudit of dimension $\dhilb=O(1)$ on each site of the graph, and with some graph metric, and each $O_i$ is supported on a set of diameter $O(1)$.  There is no requirement that $H$ be redundancy-free.

We wish to prepare a thermofield double state.  This is defined as follows.  Consider two copies of the system, calling them the left and right copies.
We define the thermofield double to be
\be
\psitf(\beta) \equiv Z^{-1} \Bigl(\exp(-\beta H/4)_L \otimes \exp(-\beta \overline H/4)_R\Bigr) \Psi_{ME}
 \ee
where $Z$ is a normalization, where the tensor product is between left and right copies, and where $\Psi_{ME}$ is a maximally entangled state given by, up to normalization, $\sum_n |n\rangle_L \otimes |n\rangle_R$, where $|n\rangle$ is any complete basis of states.  Here $\overline H$ is the complex conjugate (or transpose) of $H$ in the given basis.
Indeed, this is equivalent to 
$\psitf(\beta) = Z^{-1} \Bigl(\exp(-\beta H/2)_L \otimes I \Bigr) \Psi_{ME}$.

So, for any $H=\sum_i J_i O_i$, 
\be
\label{betachange}
\partial_\beta \psitf(\beta) =
\sum_{i} (O_{i,L} - \langle O_{i,L} \rangle) \psitf(\beta),
\ee
where the expectation value is in the given thermofield double state.  We add the subscript $L$ to the operator to indicate which copy it acts on.

Suppose that for any $\epsilon>0$, for all $i$, for $\beta=O(1)$, there is some anti-unitary operator $\eta_i$ supported on sites within distance $\ell=\log_C(1/\epsilon))$ of $i$, for some fixed $C$, such that
\be
\label{erreq}
|(O_{i,L} - \langle O_{i,L} \rangle) \psitf(\beta)-\eta_i(\beta) \psitf(\beta)|\leq \epsilon.
\ee
Then, this operator $\eta_i$ acts on a Hilbert space of dimension $O(\dhilb^{d^\ell})$.

We can then
 simulate the evolution under operator $\eta(\beta)=\sum_i \eta_i(\beta)$ on a quantum computer in time polynomial in $n,\beta$, and in the total error, by regarding it as time evolution under a time-dependent Hamiltonian $i\eta(\beta)$, with $\beta$ playing the role of time.  First order Trotter-Suzuki methods suffice to show this.
Here, since the derivative $\partial_\beta\psitf(\beta)$ is bounded in norm by $O(n)$, we can assume a bound on the  derivative $\partial_\beta \eta_i(\beta)$, needed for this time-dependent evolution.

Since of course $\psitf(0)$ is easy to prepare, the time then scales at worst as 
$O(\dhilb^{d^\ell})$.  We now show that the existence of these operators follows from a certain correlation decay, \cref{corrdecayneeded}, and from an assumption of bounded $\beta$, \cref{likely}.

The correlation decay will lead to
with $\ell \sim \log_C(1/\epsilon)$ for $C \sim 1/(\beta^2 d)$ for redudancy-free Hamiltonians, and a related decay for general Hamiltonians, giving the claimed simulation.

Let us define three sets of sites, $A,B,C$.  Set $A$ will consist of the sites in the support of $O_i$.  Set $C$ will consist of sites a distance greater than $\ell$ from $i$.  Set $B$ will consist of the remaining set of sites.
Suppose that for all operators $O_C$ supported on $C$, with $\Vert O_C \Vert \leq 1$, where $O_C$ may act on both $L$ and $R$ copies in the thermofield double state, we have
\be
\label{corrdecayneeded}
\langle (O_{i,L} - \langle O_{i,L} \rangle ) O_C \rangle \leq \epsilon',
\ee
for some $\epsilon'$, where again the expectation value is in the thermofield double state $\psitf(\beta)$.

The operator $O_i$ acts on a Hilbert space of dimension $O(1)$.  Let us, for simplicity, take the case that, up to local unitary rotation, it acts on a single qubit and let $A$ denote that single qubit; the case of a larger Hilbert space dimension is similar.
Call the two states corresponding to different eigenvalues up and down, denoted as $|\uparrow\rangle,|\downarrow\rangle$.
So, given that $\psitf(\beta)$ is a pure state, we may write $\psitf(\beta)$ as $\sum_{\sigma=\uparrow,\downarrow} A(\sigma) |\sigma\rangle_A \otimes |\psi(\sigma)\rangle_{BC}$, where $A(\sigma)$ is a scalar and $|\psi(\sigma)\rangle_{BC}$ has norm $1$.   Define $\rho_C(\sigma)$ to be the reduced density matrix on $C$ of $\psi(\sigma)$.

We may assume that
\be
\label{likely}
|A(\sigma)|^2\geq C,
\ee
for some $C>0$, as $\beta$ is bounded.
From \cref{corrdecayneeded}, then, $|\rho_{BC}(\uparrow) - \rho_{BC}(\downarrow)|_1 \leq O(\epsilon)$.
So, by the Powers-St{\o}rmer inequality\cite{powers1970free}, 
\be
\label{PS}
|\rho_{C}(\uparrow)^{1/2} - \rho_{C}(\downarrow)^{1/2}|_2^2 \leq O(\epsilon),
\ee
where $|\ldots|_2$ is the Hilbert-Schmidt norm.
Here, for any positive semi-definite operator $O$, we define $O^{1/2}$ to be the non-negative Hermitian square-root of $O$.
Assume the dimension $D_C$ of the Hilbert space on $C$ equals the dimension $D_B$ of that on $B$; if it is not, then add additional states to $B$ or $C$, with $|\psi(\sigma)\rangle_{BC}$ having vanishing amplitude on these states.  Assume that the reduced density matrices $\rho_C(\sigma)$ are both full rank (if not, infinitesimally perturb $|\psi(\sigma)\rangle_{BC}$.
We can interpret $\rho_{C}(\sigma)^{1/2}$ as a state on a tensor product Hilbert space of dimension $D_C \otimes D_C$, 
and then
$\psi_{BC}(\sigma)=U(\sigma) \rho_C(\sigma)^{1/2}$, where $U$ is a unitary.
Then, by \cref{PS},
\be
\label{rotate}
|\psi(\uparrow)\rangle_{BC} -U(\uparrow) U(\downarrow)^{\dagger} |\psi(\downarrow)\rangle_{BC}|^2 \leq O(\epsilon),
\ee
where $|\ldots|$ is the usual $\ell_2$ norm in quantum mechanics.

In this case of $O_i$ having two eigenvalues, let us write $O_i=Z_i$.
Define operators $\tilde X_i=|\uparrow\rangle_A\langle\downarrow| \otimes U(\uparrow)U(\downarrow)^\dagger+{\rm h.c.}$
and $\tilde Y_i=\sqrt{-1} |\uparrow\rangle_A\langle\downarrow| \otimes U(\uparrow)U(\downarrow)^\dagger+{\rm h.c.}$
Then, $Z_i,\tilde X_i,\tilde Y_i$ obey Pauli anticommutation relations with each other, and are Hermitian,
and from \cref{rotate}, $(\tilde X_i + i \tilde Y_i)$ approximately maps $|\downarrow\rangle_A \otimes |\psi(\downarrow)\rangle_{BC}$ to 
 $|\uparrow\rangle_A \otimes |\psi(\uparrow)\rangle_{BC}$.
 
 Of course, if we did not have $BC$ and we had a single qubit on $A$ in some given state $\psi_A$, we could infinitesimally change that state to
 the state $\psi_A + (Z_i - \langle Z_i \rangle) \psi_A {\rm d}\beta$ by applying an infinitesimal unitary rotation generated by operators $X_i,Y_i$, i.e., in this way we can define $\eta_i$.
 With the operators $\tilde X_i,\tilde Y_i$, we can in the same way generate the evolution of $\psitf(\beta)$ up to error $O(\epsilon^{1/2})$ in $\ell_2$ norm.

Note: since our evolution to prepare the state is unitary, any error that we create does not grow.  If instead we created an approximate state, using this unitary evolution, and then evolved it under a non-unitary evolution \cref{betachange}, the error might grow.

\section{Towards Improved Scaling}
\label{sec:beyond}
Here we give a possible route to algorithms polynomial in $n,\epsilon$.  We will start by considering general Hamiltonians (not necessarily redundancy-free) defined on an interaction graph.  We will have some graph of degree at most $G$ and let $H=\sum_{(i,j)} O_{i,j}$, where the sum is over edges of the graph and where $\Vert O_{i,j}\Vert \leq 1$.  If these $O_{i,j}$ are proportional to products of Paulis then this leads to an anticommutation graph of degree at most $2(d-1)$.
We assume qudits on each site, with Hilbert space dimension $D=O(1)$.

For each $i$, define $P_i$ to project orthogonal to the maximally entangled state (entangled between $L$ and $R$) on site $i$ so that $\Psi_{ME}$ is the ground state of the ``Hamiltonian" $\sum_i P_i$.
Define a non-Hermitian operator $P_i$ by
$$P_i(\beta)=\exp(-\beta H_L/2) P_i \exp(\beta H_L/2).$$
Define a Hamiltonian, that we call the Feiguin-Klich Hamiltonian\cite{feiguin2013hermitian} by
\be
H^{FK}=\sum_i \hfk_i(\beta),
\ee
where we define
\be
\hfk_i(\beta) \equiv P_i(\beta)^\dagger P_i(\beta).
\ee

Remark: while this is the definition used in Ref,.~\onlinecite{feiguin2013hermitian}, it is possible that we could find some constant-factor improvement in the allowed by $\beta$ by defining instead $P_i(\beta)=\exp[-\beta (H_L+\overline H_R)/4] P_i \exp[\beta (H_L+\overline H_R)/4]$.

We can expand the operator $P_i(\beta)$ in a power series in $\beta$.  The $k$-th order term in the series can be
further expanded as a sum of terms proportional to commutators, each of which is a commutator $[\ldots[P,O_{i_1,j_1}],O_{i_2,j_2}],\ldots O_{i_k,j_k}]].$
These commutators define clusters which are lattice animals of size at most $k$.  
For a given \emph{multiset} of $k!$ edges $(i_a,j_a)$, there are at most $k!$ corresponding sequences and there is a factor $1/k!$ from the series for the exponential, so the sum of norm of terms corresponding to a given set is bounded by $O(\beta)^k$.
The number of multisets of size at most $k$ is $O(d)^k$.
Hence, for $\beta \lesssim 1/d$, the error in norm in truncating the series for $P_i(\beta)$ at a given $k$ is exponentially small in $k$
and a given $\hfk_i(\beta)$ can be implemented on a quantum computer (i.e., unitary evolution under this operator for small time) can be implemented with complexity exponential in $k$.
To explain this in more detail: we need to take $k\lesssim \log(\epsilon)$.  
The number of multisets at this order is $O(d)^k$, and hence polynomial in $\epsilon$.  
However, the number of sequences for a given multisets is factorial in $k$ and hence would be superpolynomial in $\epsilon$.
This would lead to a superpolynomial time if we implement each such sequence separately.  However, we can resolve this using a similar method as in \cref{compcalc}.  Let $\gamma$ denote a set (not multiset) of terms in the Hamiltonian, and let $H(\gamma)$ denote the sum of those terms.
The sum $[P,H(\gamma)]+(1/2!) [[P,H(\gamma)],H(\gamma)]+(1/3!)[[[P,H(\gamma)],H(\gamma)],H(\gamma)]+\ldots$ includes the sum of all terms 
in the expansion $P_i(\beta)$ corresponding to sequences where every Hamiltonian term in the sequence is supported in $\gamma$.
Call this sum $S(\gamma)$.
Using the same inclusion-exclusion trick as in \cref{compcalc}, the sum
$\sum_{\gamma'\subset \gamma} (-1)^{|\gamma|-|\gamma'|} S(\gamma')$ equals the sum of all sequence where every 
Hamiltonian term in the sequence is supported in $\gamma$ and where each term in $\gamma$ appears at least once in the sequence.
Hence, the sum over $\sum_\gamma S(\gamma)$ is the desired expansion for $P_i(\beta)$ and we may truncate at $\gamma$ with $|\gamma|\lesssim \log(\epsilon)$.  For each such $\gamma$, we can compute $S(\gamma)$ with complexity exponential in $|\gamma|$ and hence polynomial in $|\epsilon|$.
So,
we can simulate evolution under $H^{FK}$ on a quantum computer to accuracy $\epsilon$ in time polynomial in $n,\epsilon$ for $\beta \lesssim 1/d$.

The ground state of $H^{FK}$ is the desired thermofield double state, with zero energy.
If we could show that $H^{FK}$ has a spectral gap, then we can prepare $\Psi_{ME}$ and then adiabatically change $\beta$ in the definition of $H^{FK}$ from the $\beta=0$ to some desired final $\beta$.  The adiabatic evolution can be implemented in various ways\cite{farhi2000quantum,aharonov2003adiabatic,boixo2009eigenpath}

The next subsection is devoted to proving
\begin{theorem}
\label{gapthm}
For sufficiently small $\beta d$, the Hamiltonian $H^{FK}$ has a unique ground state and a spectral gap $\gamma$ strictly bounded away from zero, with $\gamma$ approaching $1$ as $\beta d$ tends to zero.
\end{theorem}

\subsection{Spectral Gap}
To show spectral gap, consider first the spectrum of $\hfk_i(\beta)$ for a given $i$.
We claim:
\begin{lemma}
\label{hfkspec}
For sufficiently small $\beta d$, the operator $\hfk_i(\beta)$ has the same number of zero eigenvalues as $P_i$ does, and the nonzero eigenvalues of $\hfk_i(\beta)$ lie in an interval around $1$, with the width of the interval tending to zero as $\beta d$ tends to zero.
\begin{proof}
While this operator is Hermitian, we apply a similarity transform which does not change the eigenvalues
and consider the spectrum of the non-Hermitian operator
$$
\exp(\beta H_L) P_i \exp(-\beta H_L) P_i.
$$
This operator is block upper triangular, if the blocks corresponding to the subspaces $P_i=0$ and $P_i=1$.  So, it has at least as many zero eigenvalues as $P_i$ does, and to compute the remaining spectrum we may consider the operator on the diagonal block where $P_i=1$.
In this block, the operator is close in operator norm to $P_i$ which is equal to $1$ on this block. Here ``close" means that for any desired bound $\delta>0$ on the norm, this bound holds for sufficiently small $\beta d$.  Since $P_i$ is Hermitian, by the Bauer-Fike theorem all nonzero eigenvalues of $P_i(\beta)^\dagger P_i(\beta)$ lie in some (arbitrarily small) interval near $1$.
\end{proof}
\end{lemma}

In particular, the lemma implies that the nonzero eigenvalues of $\hfk_i(\beta)$ are lower bounded by a constant times $\hfk_i(\beta)^2$.

Remark: indeed, we can tighten this and show that the nonzero eigenvalues of $\hfk_i(\beta)$ are at least $1$.  This is a general result, which is not needed and this paragraph may be skipped.  Consider any projector $\Pi$ and any Hermitian invertible matrix $X$.  Consider $M=(X\Pi X^{-1})^\dagger (X^{-1}\Pi X)=X \Pi X^{-2} \Pi X$.  Let $\Pi$ have rank $r$ so $M$ has rank at most $r$.  Write $X$ as a block two-by-two matrix with the blocks corresponding to the range of $\Pi$ and the orthogonal complement of that range,
$$X=\begin{pmatrix} A & B \\ B^\dagger & C\end{pmatrix}.$$
The upper-left block of $X^2$ equals $A^2+B B^\dagger$.
Using formulas for inverse of block matrices, the upper-left block of $X^{-2}$ is some matrix $Q$ with $Q\geq (A^2+B^\dagger B)^{-1}$ where the inequality for matrices means $B-A$ is positive semi-definite.  Hence,
$$M\geq \begin{pmatrix} A\\ B^\dagger \end{pmatrix} \cdot (A^2+B B^\dagger)^{-1} \cdot \begin{pmatrix} A & B\end{pmatrix},$$
and one may verify by squaring that the right-hand side is a projector, also of rank $r$, so $M$ is lower bounded by a projector of the same rank as $M$ and hence all nonzero eigenvalues of $M$ are at least $1$.

The next step is similar to Ref.~\onlinecite{fannes1992finitely} in that
to show a spectral gap $\gamma$ for $H^{FK}$
it suffices to show that
$(H^{FK})^2 \geq \gamma H^{FK}$.
We have
$$(H^{FK})^2 =\sum_i \hfk_i(\beta)^2+\sum_{i\neq j} \hfk_i(\beta) \hfk_j(\beta)$$ and using the spectrum of
$\hfk_i(\beta)$, we have that $\hfk_i(\beta)^2\geq \hfk_i(\beta)$.  
Here we use the remark above lower bounding the eigenvalues of $\hfk_i(\beta)$ by $1$; even just using \cref{hfkspec} suffices to lower bound $\hfk_i(\beta)^2$ by a constant times $\hfk_i(\beta)$ which is also sufficient for us.
So, we now lower bound
$\hfk_i(\beta) \hfk_j(\beta)$ which is not necessarily positive semi-definite.

Consider any pair $i \neq j$.  Here these $i,j$ are \emph{not} necessarily neighbors.
We will lower bound the operator
$\hfk_i(\beta) \hfk_j(\beta)$.
Indeed, $\hfk_i(\beta) \hfk_j(\beta)=P_i(\beta)^\dagger \hfk_j(\beta) P_i(\beta) + P_i(\beta)^\dagger [P_i(\beta),\hfk_j(\beta)]\geq 
P_i(\beta)^\dagger [P_i(\beta),\hfk_j(\beta)]$.
Hence, 
\be
\label{Cbound}
\hfk_i(\beta) \hfk_j(\beta)\geq -\Vert P_i(\beta)^\dagger [P_i(\beta),\hfk_j(\beta)] \Vert.
\ee

Next consider the spectrum of $\hfk_i(\beta)+\hfk_j(\beta)$.  
\begin{lemma}
\label{hiplushjbound}
For $i\neq j$,
for sufficiently small $\beta d$, the operator $\hfk_i(\beta)+\hfk_j(\beta)$ has the same number of zero eigenvalues as $P_i P_j$ does, and the nonzero eigenvalues of $\hfk_i(\beta)+\hfk_j(\beta)$ are lower bounded by 
some quantity $\Delta$ that approaches $1$ as $\beta d$ tends to zero.
\begin{proof}
While this operator is Hermitian, we apply a similarity transform which does not change the eigenvalues
and consider the spectrum of the non-Hermitian operator
$$
\exp(\beta H_L) P_i \exp(-\beta H_L) P_i+
\exp(\beta H_L) P_j \exp(-\beta H_L) P_j.
$$
This operator is block upper triangular, if the two blocks correspond to the block where $P_i=P_j=0$ and the block where $P_i+P_j>0$.  Hence, it has at least as many zero eigenvalues as $P_i+P_j$ does and on the remaining block we can again regard it as a non-Hermitian perturbation of a Hermitian operator and apply Bauer-Fike to lower bound it by a constant $\Delta$ which may be taken arbitrarily close to $1$ for $\beta$ sufficiently small compared to $1/d$.
So, indeed, it has the same number of zero eigenvalues as $P_i+P_j$ does.
\end{proof}
\end{lemma}

Next we show
\begin{lemma}
For $i\neq j$, we have
\be
\hfk_i(\beta) \hfk_j(\beta) \geq -(1/\Delta)  \Vert P_i(\beta)^\dagger [P_i(\beta),\hfk_j(\beta)] \Vert    (\hfk_i(\beta)+\hfk_j(\beta)),
\ee
where $\Delta$ is from \cref{hiplushjbound}.
\begin{proof}
Let ${\cal H}_0$ denote the subspace of the Hilbert space on which $\hfk_i(\beta)+\hfk_j(\beta)=0$.  Let ${\cal H}^\perp$ denote the orthogonal subspace.
Write $\hfk_i(\beta) \hfk_j(\beta)$ as a block two-by-two matrix, with the blocks corresponding to ${\cal H}_0$ and ${\cal H}^\perp$.
This matrix vanishes every except in the lower-right block, where it is lower bounded by $\Vert P_i(\beta)^\dagger [P_i(\beta),\hfk_j(\beta)] \Vert$ by \cref{Cbound}.
Writing $\hfk_i(\beta)^ + \hfk_j(\beta)$ also as a block two-by-two matrix, it also vanishes everywhere except on the lower-right block where it is lower bounded by $\Delta$.
\end{proof}
\end{lemma}

Next we bound $\Vert P_i(\beta)^\dagger [P_i(\beta),\hfk_j(\beta)] \Vert $, and finally we use that to prove \cref{gapthm}.

We have
\begin{lemma}
For sufficiently small $\beta d$, we have
$$ \Vert P_i(\beta)^\dagger [P_i(\beta),\hfk_j(\beta)] \Vert =\sum_{\ell>0} n(i,j,\ell) O(\beta)^{\ell},$$ where $n(i,j,\ell)$ is the number of distinct paths from $i$ to $j$ with distance $\ell$ with the graph metric.
\begin{proof}
To bound the norm  
$ \Vert P_i(\beta)^\dagger [P_i(\beta),\hfk_j(\beta)] \Vert $,
write each given $\hfk$ in terms of $P_i$ or $P_j$.  So, there are a total of \emph{four} different factors of $P$, i.e., one each of $P_i,P_i^\dagger,P_j,P_j^\dagger$.
Expand each $P$ in a series in $\beta$, which as above, 
can be written as a sum of terms on connected multisets, with the norm of the term bounded by the
$O(\beta)^k$ where $k$ is the cluster multisets.
The commutator is only nonvanishing if some multiset in an expansion of $P_i$ intersects some multiset in an expansion of $P_j$ or $P_j^\dagger$.
So, the nonvanishing terms can be written as a single connected multiset, as well as, for each element of the multiset, a number from $1$ to $4$ indicating which of the different $P_i$ or $P_j$ it contributes to the expansion of.  This choice of a number from $1$ to $4$ increases the norm by only $4^k$ if the multiset is of size $k$.  For each $\ell>0$, consider the sum of all terms where the minimum length path from
$i$ to $j$ is $\ell$.  There are $n(i,j,\ell)$ possible choices of that minimum length path; order those paths arbitrarily.  For each such path $P$, consider the terms containing that path $P$ and no previous path in the given order; the sum of the norm of those terms is bounded by $O(\beta)^{\ell}$ for sufficiently small $d\beta$ as the number of connected multisets containing $P$ with size at most $P+m$ is bounded by $O(d)^m$.
\end{proof}
\end{lemma}

We can now show that $(H^{FK})^2 \geq \gamma H^{FK}$ for sufficiently small $\beta d$, proving \cref{gapthm}.
Indeed, then
\begin{eqnarray}
\label{thegap}
(H^{FK})^2 &=& \sum_i \hfk_i(\beta)^2+\sum_{i\neq j} \hfk_i(\beta)^2 \hfk_j(\beta)^2 \\ \nonumber 
&\geq &H^{FK}-  (1/\Delta) \sum_i \Bigl(\sum_j \Vert P_i(\beta)^\dagger [P_i(\beta),\hfk_j(\beta)] \Vert \Bigr)     \hfk_i(\beta).
\end{eqnarray}
Then, if we have the given bound on  $ \Vert P_i(\beta)^\dagger [P_i(\beta),\hfk_j(\beta)] \Vert $ then for $\beta$ sufficiently small compared to $1/d$, the sum 
$\sum_j \Vert P_i(\beta)^\dagger [P_i(\beta),\hfk_j(\beta)] \Vert$
will be bounded by $\sum_j \sum_{\ell>0} n(i,j,\ell) O(\beta)^{\ell}=\sum_{\ell>0} \sum_j n(i,j,\ell) O(\beta)^{\ell}\leq \sum_{\ell>0} d^{\ell} O(\beta)^{\ell}$ which for
small enough $d\beta$
will be bounded by a constant which we may take arbitrarily small, and so the second term on the right-hand side of \cref{thegap} is lower bounded by an arbitrarily small constant times $-H^{FK}$.

\subsection{Redundancy-Free}
We now discuss the redundancy-free case, for $\beta \lesssim 1/\sqrt{d}$.

First let us show that we can truncate the power series series expansion of $P_i(\beta)$ with minimal error \emph{in the energy of the ground state of} $H^{FK}$.
That is, we will seek a finite approximation to $\hfk_i(\beta)$ such that $\langle \Psi_{tf} | \hfk_i(\beta) | \Psi_{tf}\rangle$ converges to zero exponentially in the order of our expansion.

As before expand the operator $P_i(\beta)$ in a power series in $\beta$.  We now write this as a sum over terms, where each term is specified by a lattice animal $a$ and a multiplicity $m_j$ for each vertex $j$ in the animal.
The term for given $a$ and vector $\vec m$ of multiplicities will include all terms where 
each operator $O_j$ appears $m_j$ times.  As there are at most $(\sum_j m_j)!$ possibly orderings of these terms, and there is a factor $1/(\sum_j m_j)!$ in the expansion of the exponential, the norm of a term is bounded by $O(\beta)^{\sum_j m_j}$.

Consider truncating the series at given $k$, where by truncation we mean the sum over all terms with $|a|\leq k$, regardless of whether the total multiplicity is greater than $k$.
This is still an infinite number of terms, as one must sum over multiplicities, however, since the norm decays exponentially in $\sum_j m_j$, we can approximate this sum to exponentially small error by also truncating in the total multiplicity.  

  However, since we now assume only that
 $\beta \lesssim 1/\sqrt{d}$,  and since there are $O(d)^{|a|}$ animals $a$ with cardinality $|a|$,
 we cannot bound the error in operator norm by truncating the series at a given $k$.
 Instead, we show that if $\hfk_i(\beta,k)$ is defined to be the truncation of the series for $\hfk_i$ at given $k$, then
 $\langle \Psi_{tf} | \hfk_i(\beta,k)^2 | \Psi_{tf}\rangle$ is exponentially small in $k$ for $\beta$ sufficiently small compared to $1/\sqrt{d}$.
 Note here that $\hfk_i(\beta)$, being a product of $P_i^\dagger(\beta)$ and $P_i(\beta)$ which have power series expansions with the given norm bound, also has a power series expansion as a sum of animals with multiplicities, with the given norm bound; given an animal $a_1$ in the expansion of $P_i^\dagger(\beta)$ with multiplicity $\vec m_1$, and an animal $a_2$ in the expansion of $P_i(\beta)$ with multiplicity $\vec m_2$, their union is also an animal with multiplicity equal to the sum of the multiplicities, and, at a factor of at most $2^k$, we can label for each vertex in the union animal, whether it came from $a_1$ or $a_2$.
  
  Define $\overline \hfk_i(\beta,k)$ to equal $\hfk_i(\beta)-\hfk(\beta,k)$ so that it equals the error in truncating at order $k$.
 Then, since $\hfk_i(\beta) | \Psi_{tf}\rangle=0$, we have
  $\langle \Psi_{tf} | \hfk_i(\beta,k)^2 | \Psi_{tf}\rangle=
  \langle \Psi_{tf} | \overline \hfk_i(\beta,k)^2 | \Psi_{tf}\rangle$.
The operator $\overline \hfk_i(\beta,k)$ is a sum over animals $a$ with $|a|>k$ and over multiplicites for each animal, and similarly
$\overline \hfk_i(\beta,k)^2$ is a sum over animals $a$ with $|a|>2k$ and over multiplicities.
For each such animal $a$, specify for each site $j$ in the animal whether $m_j$ is even or odd.
So, for each animal $a$, we have $2^{|a|}$ terms, each term being a sum of all possible multiplicities with the given parity of each $m_j$.
For a given term, let $n_{\rm odd}$ be the number of sites with odd multiplicities and $n_{\rm even}$ be the number of those with even multiplicities.
The norm of a given term is bounded by $O(\beta)^{n_{\rm odd}+2 n_{\rm even}}$.
Since there are $O(d)^{|a|}$ animals with given $|a|$, if we can bound the expectation value of a given animal by $O(\beta)^{n_{\rm odd}}$ times the norm of the given term,
it will bound the 
sum over all animals with $|a|>2k$ by $O(\kappa)^{2k}$ and so for sufficiently small $\kappa$ this will decay exponentially in $k$.

To bound the expectation of an animal with given parities, again use the cluster expansion as was used for the partition function.  Allow also in the definition of polymers for the partition function not just terms in the exponential of the Hamiltonian but also terms where the animal with given parity appears, with some prefactor $\delta$ for those terms; the polymer $\gamma$ must contain $a$ for the given animal $a$ to appear.  We then compute the expectation value by taking the derivative of the log partition function with respect to
$\delta$.  One may verify that the weight of a polymer $\gamma$ is bounded by as $(1+\delta) O(\beta^2)^{|\gamma|}$.  The term independent of $\delta$ is the previous bound on the weight of a polymer while the term proportional to $\delta$ is still bounded by $O(\beta^2)^{|\gamma|}$ occurs because the norm is bounded by $O(\beta)^{n_{\rm odd}+2 n_{\rm even}}$ and for even site with odd parity, there must be at least one appearance of the corresponding $O_i$ which gives an additional factor of $O(\beta)^{n_{\rm odd}}$.

We conjecture that $H^{FK}$ has a spectral gap for $\beta\lesssim 1/\sqrt{d}$, but we do not have a proof of this.  We do still
 have that all nonzero eigenvalues of $\hfk_i(\beta)$ are at least $1$, so $\hfk_i(\beta)^2\geq \hfk_i(\beta)$, but establishing the gap becomes more difficult.
 
 It seems likely that one way to establish a spectral gap would be to prove the convergence of a cluster expansion for $H^{FK}$, as was done in Ref.~\onlinecite{yarotsky2006ground,klich2010stability}.  This becomes more complicated as $H^{FK}$ is not strictly local.  We leave this for future work.

\section{QMA-Hardness of the Ground Energy}
\label{sec:QMA_hard}
Estimating the ground energy remains QMA-complete for redundancy-free Hamiltonians, even with bounded locality and anticommutation-degree.

\begin{theorem}
\label{thm:qma_ground_energy}
The local Hamiltonian problem is QMA-complete for redundancy-free Pauli Hamiltonians with terms of locality at most five, coefficients of magnitude at most one, and anticommutation-degree at most $26$.
\end{theorem}

We reduce from the antiferromagnetic XY Hamiltonian on subgraphs $G=(V,E)$ of the triangular lattice,
\be
K_G=\sum_{uv\in E}J_{uv}(X_uX_v+Y_uY_v),
\qquad 0<J_{uv}\leq1.
\label{eq:qma_source}
\ee
The ground energy problem for this Hamiltonian is QMA-complete by Piddock and Montanaro~\onlinecite{piddock2017complexity}. Their polynomially bounded couplings can be rescaled to the normalization in \cref{eq:qma_source}, leaving an inverse-polynomial promise gap. Write $N=|V|$. The interaction graph has degree at most six, so $|E|\leq3N$.

Of course, this Hamiltonian need not be redundancy-free. For example, $(X_uX_v)(X_vX_w)(X_wX_u)=I$ around a triangle. We will replace each source qubit by a small block of physical qubits, and arrange that the interactions between blocks reproduce $K_G$ in an encoded subspace. The physical Pauli terms will have no relations, even though the encoded interactions do.

\subsection{Removing Relations}

Recall the simple way of removing relations mentioned in the introduction. Give each Pauli term $P_j$ its own ancilla $f_j$, and replace it by $P_jZ_{f_j}$. Any product in which that term appears an odd number of times contains a $Z$ on its ancilla, and hence cannot be scalar. We call these ancillas \emph{flags}.

However, this operation does not in general preserve the ground energy. Each $Z_{f_j}$ commutes with the Hamiltonian, and in a sector of fixed flag values it simply chooses the sign of the corresponding interaction. The ground energy is therefore the minimum over all these sign choices, rather than the energy for the signs we intended. One might try to add fields $-Z_{f_j}$ to fix the flags, but this reintroduces the original relations: multiply a relation among the original terms, now with their flags attached, by the corresponding flag fields.

There is a way around this problem. Suppose that the interaction with flag $f_j$ is $-w_jM_jZ_{f_j}$, where $w_j\geq0$ and $M_j\succeq0$. Changing that flag from $+1$ to $-1$ adds $2w_jM_j\succeq0$ to the Hamiltonian. Changing any set of flags therefore adds a positive operator. Thus the all-$(+1)$ sector has the lowest ground energy, without any fields to fix the flags.

A nonidentity Pauli operator is not positive. Its compression to a subspace, however, can be positive. We will use a three-qubit block whose ground space encodes one qubit, and choose physical coupling operators that become positive when compressed to that space. Their off-diagonal matrix elements will give the XY interaction.

\subsection{A Three-Qubit Encoding}

Take two qubits $a,b$, with polarization axes
\be
O_a=\frac{X_a+Y_a+Z_a}{\sqrt3},
\qquad
O_b=\frac{X_b+Z_b}{\sqrt2}.
\label{eq:qma_axes}
\ee
These operators commute and square to the identity. Write $|a_\pm\rangle$ and $|b_\pm\rangle$ for their eigenvectors, and consider the two states $|a_+b_+\rangle$ and $|a_-b_-\rangle$. We will couple blocks using the six Pauli operators $\sigma_a^\mu X_b$ and $\sigma_a^\mu Z_b$, where $\mu\in\{x,y,z\}$.

Both expectation values change sign between the two states, so their product does not. Each $\sigma_a^\mu$ has diagonal matrix elements $\pm1/\sqrt3$ in the $|a_\pm\rangle$ basis, and off-diagonal magnitude $\sqrt{2/3}$. For $X_b$ and $Z_b$, the corresponding quantities are $\pm1/\sqrt2$ and $1/\sqrt2$. Hence each two-qubit coupling operator has both diagonal entries equal to $1/\sqrt6$, while its off-diagonal magnitude is $1/\sqrt3$. The diagonal part has the sign we want, but the off-diagonal part is too large for the compressed operator to be positive.

We can reduce the off-diagonal entries without changing the diagonal ones by adding a third qubit $c$. Suppose that we attach normalized states $|c_+\rangle$ and $|c_-\rangle$ to the two outer-qubit states, with real positive overlap $\chi=\langle c_+|c_-\rangle$. For any operator acting only on $a,b$, its matrix on these two states changes as
\be
\begin{pmatrix}q_{00}&q_{01}\\q_{10}&q_{11}\end{pmatrix}
\longmapsto
\begin{pmatrix}q_{00}&\chi q_{01}\\\chi q_{10}&q_{11}\end{pmatrix}.
\label{eq:qma_overlap}
\ee
Here the matrix on the left is in the basis $|a_+b_+\rangle,|a_-b_-\rangle$, and on the right we have included the central-qubit states. Thus any $0<\chi<1/\sqrt2$ makes the six compressed coupling operators positive, while leaving their off-diagonal entries nonzero.

Couple both outer polarizations to $Z_c$, and add a transverse field on $c$:
\be
h=-X_c-Z_c(X_a+Y_a+Z_a)-Z_c(X_b+Z_b).
\label{eq:qma_block}
\ee
This is a sum of six Pauli terms, all with coefficient $-1$. The $Z_c$ couplings favor the two aligned choices of the outer polarizations. The transverse field will give the corresponding central-qubit states a nonzero overlap.

Let $|c_\pm\rangle$ now be the ground states of $-X_c\mp(\sqrt3+\sqrt2)Z_c$, with phases chosen so that their overlap is positive. Define
\be
R=\sqrt{6+2\sqrt6},
\qquad
r=\sqrt{6-2\sqrt6},
\qquad
\chi=\langle c_+|c_-\rangle=\frac1R.
\label{eq:qma_constants}
\ee
The two encoded basis states are
\be
|0_L\rangle=|a_+\rangle|b_+\rangle|c_+\rangle,
\qquad
|1_L\rangle=|a_-\rangle|b_-\rangle|c_-\rangle.
\label{eq:qma_code}
\ee
Let $P$ project onto their span. Use the same basis choices in every block, and write $\overline X,\overline Y,\overline Z$ for the usual Pauli operators on this encoded qubit. In particular, $\overline X=|0_L\rangle\langle1_L|+|1_L\rangle\langle0_L|$. Denote the compressed coupling operators by
\be
M_\mu^X=\left.P(\sigma_a^\mu X_b)P\right|_{\operatorname{im}P},
\qquad
M_\mu^Z=\left.P(\sigma_a^\mu Z_b)P\right|_{\operatorname{im}P}.
\label{eq:qma_ports}
\ee
\begin{lemma}
\label{lem:qma_encoding}
The block Hamiltonian has a two-dimensional ground space equal to the code, and
\be
hP=-RP,
\qquad
h+RI\succeq(R-r)(I-P)\succeq I-P.
\label{eq:qma_block_gap}
\ee
Every $M_\mu^X$ and $M_\mu^Z$ is positive definite, with eigenvalues
\be
\frac1{\sqrt6}\pm\frac{\chi}{\sqrt3}.
\label{eq:qma_port_eigenvalues}
\ee
For two blocks,
\be
\sum_{\mu=x,y,z}M_\mu^X\otimes M_\mu^Z
=\frac12I-\frac{\chi^2}{2}
(\overline X\otimes\overline X+\overline Y\otimes\overline Y).
\label{eq:qma_port_identity}
\ee
\end{lemma}

\begin{proof}
Since $h$ commutes with $O_a$ and $O_b$, we can fix their eigenvalues $\sigma,\tau\in\{\pm1\}$. The central-qubit Hamiltonian in this sector is
\be
-X_c-(\sqrt3\sigma+\sqrt2\tau)Z_c.
\ee
When the signs are aligned, its eigenvalues are $\pm R$; when they are opposite, they are $\pm r$. Each eigenvalue occurs twice in the full block. Thus the two states in \cref{eq:qma_code} span the ground space, with energy $-R$, and the next energy is $-r$. Since $R-r>1$, this proves \cref{eq:qma_block_gap}.

The Bloch vectors of $|c_\pm\rangle$ are $(1,0,\pm(\sqrt3+\sqrt2))/R$. Their overlap is therefore $\chi=1/R$, as claimed. From the matrix elements computed above and \cref{eq:qma_overlap}, each compressed coupling operator has diagonal entries $1/\sqrt6$ and off-diagonal magnitude $\chi/\sqrt3$. This gives \cref{eq:qma_port_eigenvalues}, and both eigenvalues are positive because $\chi<1/\sqrt2$.

To see how the XY interaction arises, first sum over the three Pauli axes on the $a$ qubits. We have
\be
\sum_{\mu=x,y,z}\sigma^\mu\otimes\sigma^\mu
=2\,\mathrm{SWAP}-I.
\ee
This identity is unchanged by applying the same basis change to both qubits, so it also holds in the $|a_\pm\rangle$ basis. On a logical basis state $|s_1s_2\rangle_L$, with $s_1,s_2\in\{0,1\}$, the diagonal entry of the $a$-part is $(-1)^{s_1+s_2}$. The diagonal entry of the $X_b\otimes Z_b$ part is $(-1)^{s_1+s_2}/2$. Thus all four diagonal entries of the compressed sum are $1/2$.

The only possible off-diagonal transition is $|01\rangle_L\leftrightarrow|10\rangle_L$, since the $a$-part only exchanges the two $a$ states. Its matrix element is $2$ from the $a$-part, multiplied by
\be
\langle b_-|X|b_+\rangle\langle b_+|Z|b_-\rangle
=-|\langle b_-|X|b_+\rangle|^2=-\frac12
\ee
from the $b$-part. Here we used that $X+Z=\sqrt2 O_b$ is diagonal in the $|b_\pm\rangle$ basis, so the off-diagonal entries of $X$ and $Z$ are negatives of one another. Finally, the two central qubits contribute a factor $\chi^2$. The off-diagonal matrix element is therefore $-\chi^2$, which proves \cref{eq:qma_port_identity}.
\end{proof}

By \cref{eq:qma_port_identity}, a negative physical coupling gives an antiferromagnetic XY interaction on the encoded qubits, together with a scalar shift.

\subsection{Coupling the Encoded Qubits}

Now put one block at each vertex of $G$. Choose an arbitrary orientation $u\to v$ for each edge; this only specifies which end uses $X_b$ and which uses $Z_b$. For each $\mu\in\{x,y,z\}$, introduce a separate flag $f_{uv,\mu}$, and set
\be
W_{uv}=-J_{uv}\sum_{\mu=x,y,z}
\sigma_{a_u}^\mu\sigma_{a_v}^\mu
X_{b_u}Z_{b_v}Z_{f_{uv,\mu}}.
\label{eq:qma_edge}
\ee
The Hamiltonian we construct is
\be
H=H_0+\eta W,
\qquad
H_0=\sum_{v\in V}h_v,
\qquad
W=\sum_{uv\in E}W_{uv},
\label{eq:qma_output}
\ee
where $\eta>0$ will be chosen small enough that exciting a block cannot change the ground energy by much. Define
\be
S=\sum_{uv\in E}J_{uv},
\qquad
L=\max(1,3S),
\qquad
\lambda=\frac{\chi^2}{2}=\frac1{2R^2}>\frac1{24}.
\label{eq:qma_parameters}
\ee
There are three Pauli terms per edge, so $\|W\|\leq3S\leq L$.

Let $P_v$ be the code projector for block $v$, and let $P_0$ project onto the product of these codes, tensored with the identity on the flags. Fix the flag values $z_{uv,\mu}\in\{\pm1\}$, and call the corresponding restriction of $P_0WP_0$ $W_{\rm enc}(z)$. When all flags are $+1$, \cref{eq:qma_port_identity} gives
\be
W_{\rm enc}(+)=\lambda\overline K_G-\frac S2I,
\ee
where $\overline K_G$ is the source Hamiltonian acting on the encoded qubits.

For general flag values,
\be
\begin{split}
W_{\rm enc}(z)
={}&\lambda\overline K_G-\frac S2I\\
&+2\sum_{\substack{uv\in E,\ \mu\in\{x,y,z\}\\z_{uv,\mu}=-1}}
J_{uv}M_\mu^{X,(u)}\otimes M_\mu^{Z,(v)}.
\end{split}
\label{eq:qma_flag_order}
\ee
Each term in the second line is positive, since its factors are positive operators on different blocks. Hence $W_{\rm enc}(z)\succeq W_{\rm enc}(+)$ for every flag configuration. Writing $E_0(A)$ for the ground energy of a Hermitian operator $A$, we obtain
\be
e_{\rm enc}:=E_0\left(\left.P_0WP_0\right|_{\operatorname{im}P_0}\right)
=\lambda E_0(K_G)-\frac S2.
\label{eq:qma_encoded_energy}
\ee
\subsection{Redundancy-Freeness and Locality}

Every edge term has a flag on which no other term acts. Thus a product of distinct Pauli terms that is scalar cannot contain any edge term. It remains to consider relations within a single block, since different blocks have disjoint supports.

The six terms of $h$ also have no relation. The term $X_c$ is the only one containing $X$ on $c$, so it cannot occur in a scalar product. The two terms $Z_cX_b$ and $Z_cZ_b$ are then excluded by their action on $b$: neither one, nor their product, is scalar on that qubit. Among the remaining three terms $Z_cX_a,Z_cY_a,Z_cZ_a$, a nonempty product can be scalar on $a$ only if it contains all three. But their product still contains $Z_c^3=Z_c$, so it is not scalar. Hence the only scalar product of distinct terms is the empty product. Since any Pauli product can be reduced, up to a scalar phase, by canceling repeated factors in pairs, the full Hamiltonian is redundancy-free.

There are $3N+3|E|\leq12N$ qubits and $6N+3|E|\leq15N$ terms. The block terms are at most two-local, and the edge terms are five-local. Their coefficients are $-1$ and $-\eta J_{uv}$, respectively, so all have magnitude at most one provided $\eta\leq1$.

To bound the anticommutation-degree, first consider an edge term. At each endpoint it anticommutes with two of the three terms $Z_c\sigma_a^\nu$, and with one of $Z_cX_b,Z_cZ_b$. This gives six block neighbors in total. It commutes with the other two terms on its own edge, since the two anticommutations on the $a$ qubits cancel. On another edge sharing one endpoint, it anticommutes with two of that edge's three terms if the $b$-qubit Pauli is the same, and with one if it is different. There are at most ten such edges, since $G$ has degree at most six. Thus the anticommutation-degree of an edge term is at most $6+2\cdot10=26$.

For a block term, $X_c$ has five neighbors, all within its block. A term $Z_c\sigma_a^\mu$ has three block neighbors and at most $2\cdot6=12$ edge neighbors. A term $Z_cX_b$ or $Z_cZ_b$ has two block neighbors and at most $3\cdot6=18$ edge neighbors. The degree of every block term is therefore at most $20$, and the degree of the full anticommutation graph is at most $26$.

\subsection{Ground Energy and Promise Gap}

The interaction can mix code states with excited block states. By \cref{eq:qma_block_gap},
\be
H_0+NR I\succeq\sum_{v\in V}(I-P_v)\succeq I-P_0.
\label{eq:qma_product_gap}
\ee
For the second inequality, use a common eigenbasis of the commuting block projectors. Every basis vector outside their joint image has $P_v=0$ for at least one block. Thus $H_0$ has ground energy $-NR$ and a gap of at least one above the product code.

For $\eta L\leq1/4$, a variational estimate gives
\be
\begin{split}
-NR+\eta e_{\rm enc}-2\eta^2L^2
&\leq E_0(H)\\
&\leq-NR+\eta e_{\rm enc}.
\end{split}
\label{eq:qma_energy_bound}
\ee
The upper bound follows by using an encoded ground state as a trial state. For the lower bound, write any unit vector as $\psi=p+q$, where $p=P_0\psi$, $q=(I-P_0)\psi$, and $x=\|q\|$. The component outside the code costs at least $x^2$ in unperturbed energy. The cross term from $\eta W$ has magnitude at most $2\eta Lx$, and its contribution within the excited space has magnitude at most $\eta Lx^2$. Hence
\be
\begin{split}
\langle\psi|H|\psi\rangle
&\geq-NR+x^2+\eta e_{\rm enc}(1-x^2)
-2\eta Lx-\eta Lx^2\\
&\geq-NR+\eta e_{\rm enc}
+(1-2\eta L)x^2-2\eta Lx\\
&\geq-NR+\eta e_{\rm enc}
-\frac{\eta^2L^2}{1-2\eta L}.
\end{split}
\ee
We used $|e_{\rm enc}|\leq L$ in the second line and completed the square in the third. Since $\eta L\leq1/4$, the last correction is at most $2\eta^2L^2$, proving \cref{eq:qma_energy_bound}.

Let $0<\delta\leq1$ be an inverse-polynomial lower bound on the promise gap of the normalized source problem. The separation between the two encoded promise cases is at least $\eta\lambda\delta$, while the error in \cref{eq:qma_energy_bound} is at most $2\eta^2L^2$. Thus it suffices to take $\eta L^2$ small compared to $\delta$. Choose
\be
\eta=\frac{\delta}{1000L^2}.
\label{eq:qma_eta}
\ee
Then $\eta\leq1$ and $\eta L\leq1/1000$. Substituting \cref{eq:qma_encoded_energy} into \cref{eq:qma_energy_bound} gives
\be
\begin{split}
-NR+\eta\left(\lambda E_0(K_G)-\frac S2\right)
-\frac{\eta\delta}{500}
&\leq E_0(H),\\
E_0(H)
&\leq-NR+\eta\left(\lambda E_0(K_G)-\frac S2\right).
\end{split}
\label{eq:qma_final_bound}
\ee

Suppose that the source thresholds are $a_-<a_+$ with $a_+-a_-\geq\delta$. Set
\be
\begin{split}
b_-&=-NR+\eta\left(\lambda a_--\frac S2\right),\\
b_+&=-NR+\eta\left(\lambda a_+-\frac S2\right)
-\frac{\eta\delta}{500}.
\end{split}
\label{eq:qma_thresholds}
\ee
A source YES instance gives $E_0(H)\leq b_-$, and a source NO instance gives $E_0(H)\geq b_+$. Since $\lambda>1/24$, the separation satisfies
\be
b_+-b_-\geq\eta\delta\left(\lambda-\frac1{500}\right)
>\frac{\eta\delta}{26}.
\ee
Since $L\leq\max(1,9N)$, the gap is inverse-polynomial. This proves QMA-hardness, and since containment is trivial, also \cref{thm:qma_ground_energy}.

\appendix

\section{The HDQI Expansion and Decoding Bounds}
\label{app:hdqi}
We now prove the identities for Hamiltonian DQI established in \cref{sec:hdqi_background}. 
\subsection{Ordered products}
\label{app:hdqi_expansion}

Sorting a word in the relaxed generators and using $\widehat P_i^2=I$ gives
\begin{equation}
\label{eq:hdqi_word_reduction}
\widehat P_{i_1}\cdots\widehat P_{i_k}
=\sigma_A(i_1,\ldots,i_k)\widehat P(y),
\qquad y=e_{i_1}+\cdots+e_{i_k}\in\mathbb F_2^m,
\end{equation}
where $e_i$ is the $i$th standard basis vector and
\[
\sigma_A(i_1,\ldots,i_k)
=(-1)^{\sum_{r<s:\,i_r>i_s}A_{i_ri_s}}.
\]
Each interchange of two generators contributes its commutation sign. Grouping the power series of the exponential by $y$ therefore gives
\begin{equation}
\label{eq:hdqi_coefficient_series}
a_y=\sum_{k=0}^{\infty}\frac{(-\beta/2)^k}{k!}
\sum_{\substack{i_1,\ldots,i_k\in[m]\\
e_{i_1}+\cdots+e_{i_k}=y}}
\sigma_A(i_1,\ldots,i_k)J_{i_1}\cdots J_{i_k}.
\end{equation}
Replacing $\widehat P_i$ by $v_iP_i$ gives the same reordering sign and an additional factor $v^y$. This proves \eqref{eq:hdqi_physical_expansion}. A degree-$\ell$ polynomial has the same expansion with $k\leq\ell$, so its coefficient vanishes whenever $|y|>\ell$.

For the canonical representation, distinct ordered products have distinct $X$ labels and are orthonormal for normalized trace. They also satisfy $\widehat P(y)|0^m\rangle=|y\rangle$, so
\begin{equation}
\label{eq:hdqi_imaginary_time_pilot}
|\pi_\beta\rangle
=\frac{e^{-\beta\widehat H/2}|0^m\rangle}{\sqrt{\widehat Z}}.
\end{equation}

A different redundancy-free Pauli realization has the same coefficients and admits a Clifford map to the same pilot. Suppose its generators act on $k$ qubits and write $Q_i=\omega_iX^{u_i}Z^{w_i}$, where $u_i,w_i\in\mathbb F_2^k$ and $\omega_i\in\{\pm1,\pm i\}$. With $U=(u_1\ \cdots\ u_m)$ and $V=(w_1\ \cdots\ w_m)$, their ordered products are
\begin{equation}
\label{eq:hdqi_general_pauli_phase}
Q(y)=\left(\prod_i\omega_i^{y_i}\right)
(-1)^{\sum_{i<j}(w_i\cdot u_j)y_iy_j}
X^{Uy}Z^{Vy}.
\end{equation}
A Bell basis change sends $|X^{Uy}Z^{Vy}\rangle$ to $|Uy,Vy\rangle$. Redundancy-freeness makes this binary label injective in $y$, so reversible binary elimination changes it to $|y\rangle|0^{2k-m}\rangle$. Single-qubit phase gates and controlled-$Z$ gates remove the phase in \eqref{eq:hdqi_general_pauli_phase}. Applying these operations to the TFD leaves the pilot and $2k-m$ zero qubits.

\subsection{Thermal error weights}
\label{app:hdqi_weights}

Conjugating by $Z_i$ flips the sign of $\widehat P_i$ and leaves the other canonical generators unchanged. On an ordered product it acts by $\widehat P(y)\mapsto(-1)^{y_i}\widehat P(y)$. Hence the orthogonal projection
\[
\Pi_S=\prod_{i\in S}\frac{I-\operatorname{Ad}Z_i}{2},
\qquad \operatorname{Ad}Z_i(Q)=Z_iQZ_i,
\]
selects precisely the coefficients with $S\subseteq\operatorname{supp}y$. By orthonormality,
\begin{equation}
\label{eq:hdqi_projection_probability}
\bigl\|\Pi_Se^{-\beta\widehat H/2}\bigr\|_2^2
=\widehat Z\Pr[S\subseteq\operatorname{supp}Y],
\end{equation}
where $\|Q\|_p=(\tr|Q|^p)^{1/p}$ is the normalized Schatten norm.

Interpolate the couplings indexed by $S$ over the cube $[-1,1]^S$:
\[
\widehat H(\mathbf z)
=\sum_{i\notin S}J_i\widehat P_i
+\sum_{i\in S}z_iJ_i\widehat P_i.
\]
At every corner, this Hamiltonian is conjugate to $\widehat H$ by a product of the $Z_i$. Convexity of the trace exponential then gives
\begin{equation}
\label{eq:hdqi_interpolated_partition}
\tr(e^{-\beta\widehat H(\mathbf z)})\leq\widehat Z
\qquad(\mathbf z\in[-1,1]^S).
\end{equation}
Writing each difference of endpoint values as an integral of a derivative yields
\begin{equation}
\label{eq:hdqi_finite_difference}
\Pi_Se^{-\beta\widehat H/2}
=2^{-|S|}\int_{[-1,1]^S}
\partial_S e^{-\beta\widehat H(\mathbf z)/2}\,d\mathbf z,
\end{equation}
where $\partial_S$ differentiates once in each coordinate in $S$.

Each derivative inserts $-\beta J_i\widehat P_i/2$ into the exponential. For $s=|S|$ derivatives, Duhamel's formula sums over $s!$ insertion orders, each integrated over a simplex of volume $1/s!$. For one such order, let $t_0,\ldots,t_s\geq0$ be the intervening time lengths, with $\sum_jt_j=1$. Schatten H\"older gives
\begin{align}
\label{eq:hdqi_duhamel_holder}
&\left\|e^{-\beta t_0\widehat H(\mathbf z)/2}\widehat P_{i_1}
 e^{-\beta t_1\widehat H(\mathbf z)/2}\cdots
 \widehat P_{i_s}e^{-\beta t_s\widehat H(\mathbf z)/2}\right\|_2
\notag\\
&\hspace{15mm}\leq
\prod_{j=0}^s\left\|e^{-\beta t_j\widehat H(\mathbf z)/2}\right\|_{2/t_j}
=\sqrt{\tr(e^{-\beta\widehat H(\mathbf z)})}
\leq\sqrt{\widehat Z}.
\end{align}
The generators have operator norm one, and the reciprocal H\"older exponents sum to $1/2$. The insertion orders cancel the simplex volume, giving
\[
\left\|\partial_S e^{-\beta\widehat H(\mathbf z)/2}\right\|_2
\leq\sqrt{\widehat Z}\prod_{i\in S}\frac{\beta J_i}{2}.
\]
The normalized cube integral in \eqref{eq:hdqi_finite_difference} has total mass one. Squaring its norm bound and dividing by $\widehat Z$ proves \eqref{eq:hdqi_local_stochastic}.

Put $q_i=(\beta J_i/2)^2$ and $\lambda=\sum_iq_i$. Expanding a product and applying \eqref{eq:hdqi_local_stochastic} to each subset gives, for $\theta\geq0$,
\begin{align}
\label{eq:hdqi_weight_mgf}
\mathbb E e^{\theta|Y|}
&=\mathbb E\prod_i\bigl[1+(e^\theta-1)Y_i\bigr]
\notag\\
&\leq\prod_i\bigl[1+q_i(e^\theta-1)\bigr]
\leq\exp\bigl(\lambda(e^\theta-1)\bigr).
\end{align}
For example, taking $\theta=1$ in Markov's inequality yields the tail bound
\begin{equation}
\label{eq:hdqi_weight_tail}
\Pr\!\left[|Y|\geq(e-1)\lambda+\log\frac1\eta\right]\leq\eta,
\qquad 0<\eta<1.
\end{equation}

\subsection{Robustness of coherent decoding}
\label{app:hdqi_robustness}

\begin{proof}[Proof of Theorem~\ref{thm:hdqi_robustness}]
Let $r_y=y+\operatorname{Dec}(By)$ be the residual label in \eqref{eq:hdqi_coherent_output}. Compare the normalized circuit output $|\Psi_v\rangle$ with the unnormalized target vector
\[
|\Phi_v\rangle
=\frac{|0^m\rangle|e^{-\beta H_v/2}\rangle}{\sqrt{\widehat Z}}.
\]
Their difference contains exactly the incorrectly decoded labels:
\begin{equation}
\label{eq:hdqi_output_difference}
|\Psi_v\rangle-|\Phi_v\rangle
=\frac1{\sqrt{\widehat Z}}
\sum_{y:\,r_y\neq0}a_yv^y
\bigl(|r_y\rangle-|0^m\rangle\bigr)|P(y)\rangle.
\end{equation}
The independent uniform signs satisfy $\mathbb E_v[v^yv^z]=\delta_{y,z}$. Since $\|\,|P(y)\rangle\|=1$, averaging the squared norm leaves
\begin{equation}
\label{eq:hdqi_average_vector_error}
\mathbb E_v\bigl\|\Psi_v-\Phi_v\bigr\|^2
=\frac2{\widehat Z}\sum_{y:\,r_y\neq0}|a_y|^2
=2r.
\end{equation}

The vector $|\Phi_v\rangle$ has its physical register in $|\psi_{\beta,v}\rangle$. Projecting onto the orthogonal complement of this state therefore gives
\begin{equation}
\label{eq:hdqi_projection_infidelity}
1-\langle\psi_{\beta,v}|\rho_v|\psi_{\beta,v}\rangle
\leq\bigl\|\Psi_v-\Phi_v\bigr\|^2.
\end{equation}
Taking the expectation proves \eqref{eq:hdqi_average_infidelity}. Markov's inequality bounds the infidelity by $2r/\delta$ for a $1-\delta$ fraction of sign choices. The trace distance to a pure state is at most the square root of its infidelity, so the ideal output has error at most $\sqrt{2r/\delta}$ on these signs.

An input of squared fidelity at least $1-\varepsilon_{\rm p}$ is within trace distance $\sqrt{\varepsilon_{\rm p}}$ of the ideal redundancy-free TFD. The Clifford conversion, coherent decoding, and disposal of the residual register form a quantum channel, which cannot increase trace distance. The triangle inequality proves \eqref{eq:hdqi_robustness}.
\end{proof}

For a polynomial $f$, replace $a_y$ by the coefficients of $f(\widehat H)$ and $\widehat Z$ by $\tr(f(\widehat H)^\dagger f(\widehat H))$. Equations~\eqref{eq:hdqi_output_difference}--\eqref{eq:hdqi_projection_infidelity} give the same robustness bound for the normalized vectors $|f(H_v)\rangle$, when these are nonzero.

\section{Chord expansion}
\label{app:glass_onsite}

In a compatible onsite chord diagram, each connected component of the chord-intersection graph has a single family label, $X$ or $Z$.

For a uniform pairing of $2k$ cyclically ordered endpoints, let $C_k$ be the number of connected components of the chord-intersection graph. There are $a_k=(2k-1)!!$ pairings, with $a_0=1$. Each component has two choices of family label, independently of the others. Hence, after summing the two onsite species in \cref{eq:glass_chord_measure}, the single-block partition function is
\be
F_b=\sum_{k\geq0}\frac{(b^2/2)^k}{k!}w_k,
\qquad
w_k=\mathbb E\,2^{C_k},
\qquad w_0=1.
\label{eq:glass_single_series}
\ee
The expectation is over uniform pairings. At fixed $k$, the onsite law weights each pairing by $2^{C_k}$.

\subsection{A Bound on the Number of Components}

The elementary bound $C_k\leq k$ gives only $F_b\leq e^{b^2}$. We instead use a bound on the moments of $C_k$ that is uniform in $k$.

Let $c_j$ count connected chord diagrams on $2j$ cyclically ordered endpoints, and write
\[
\mathcal C(z)=\sum_{j\geq1}c_jz^j.
\]
To count diagrams with several components, consider the component containing the first endpoint. If it has $j$ chords, its $2j$ endpoints leave $2j$ gaps. Every other component lies entirely in one of these gaps, and each gap can contain an arbitrary chord diagram. Thus, if $\mathcal A(z,u)$ counts all diagrams with a factor $u$ for each component, then
\be
\mathcal A(z,u)=1+u\mathcal C\bigl(z\mathcal A(z,u)^2\bigr).
\ee
Applying Lagrange inversion gives, for $1\leq r\leq k$,
\be
\#\{\pi:|\pi|=k,\ C_k(\pi)=r\}
=\frac1{2k+1}\binom{2k+1}{r}
[z^k]\mathcal C(z)^r.
\label{eq:glass_component_count}
\ee
Put
\[
B(z)=\sum_{j\geq0}a_{j+1}z^j.
\]
We claim that $B(z)^2\leq4B(z)$ coefficientwise. For degree $j\geq2$, the two endpoint terms in the convolution contribute $2a_{j+1}$. Each of the $j-1$ interior terms is at most $3a_j$, by log-convexity of the double factorials. Their total contribution, relative to $a_{j+1}$, is therefore at most
\[
\frac{3(j-1)}{2j+1}<\frac32.
\]
Degrees zero and one satisfy the same coefficientwise bound directly. Since all coefficients are nonnegative, we can apply the bound repeatedly. Using $\mathcal C(z)\leq zB(z)$, we obtain
\be
[z^k]\mathcal C(z)^r
\leq[z^{k-r}]B(z)^r
\leq4^{r-1}a_{k-r+1}.
\ee

Dividing by the total number $a_k$ of pairings, the remaining factors in \cref{eq:glass_component_count} satisfy
\be
\frac1{2k+1}\binom{2k+1}{r}\frac{a_{k-r+1}}{a_k}
=\frac1{r!}\prod_{j=0}^{r-2}
\frac{2k-j}{2k-2j-1}
\leq\frac{2^{r-1}}{r!}.
\ee
To see the last inequality, note that the factors in the product increase with $j$. The product of all $k-1$ factors is $2^k/(k+1)\leq2^{k-1}$. Hence every initial segment has geometric mean at most two. For $r=1$, the product is empty and equals one.

Combining the last two estimates gives
\be
\Pr(C_k=r)\leq\frac{8^{r-1}}{r!}.
\label{eq:glass_component_tail}
\ee
For every fixed $u\geq1$,
\be
\sup_{k\geq1}\mathbb E\,u^{C_k}
\leq\sum_{r\geq1}\frac{u^r8^{r-1}}{r!}
=\frac{e^{8u}-1}{8}.
\label{eq:glass_component_moment}
\ee
Taking $u=2$ in \cref{eq:glass_single_series} proves
\be
F_b\leq B_0e^{b^2/2},\qquad B_0=e^{16}.
\ee
We also have $F_b\geq e^{b^2/2}$, since retaining only one onsite family removes every compatibility constraint within that family.

\subsection{Avoiding the Large Component}

A chord avoiding the selected family must avoid its largest component.

For a uniform chord diagram, a theorem of Flajolet and Noy, stated as Theorem 3.2 of Ref.~\onlinecite{acan2017chord}, says that with probability tending to one as $k\to\infty$, the diagram consists of one large connected component and isolated simple chords. The number of these isolated chords converges in distribution to a Poisson variable of mean one. In particular, the number of chords outside a largest component is tight in probability.

For an event $E$ of uniform probability $p$, Cauchy--Schwarz gives
\be
\frac{\mathbb E[2^{C_k}\mathbf1_E]}{\mathbb E2^{C_k}}
\leq (\mathbb E4^{C_k})^{1/2}p^{1/2}.
\ee
We used $\mathbb E2^{C_k}\geq1$. The prefactor is bounded by \cref{eq:glass_component_moment}, so tightness also holds under the weighted law.

Choose one of the $2k$ endpoints uniformly as an auxiliary root, rotate it to zero, and order the others clockwise. The $2k$ circular spacings are then independent of the pairing and have the Dirichlet distribution with all parameters equal to one. Weighting by $2^{C_k}$ does not change this conditional spacing law, since $C_k$ depends only on the pairing.

As in the main text, choose a largest component using a tie rule depending only on the rooted pairing, not on its labels or spacings. Its family label defines the phase. Conditional on the uncolored pairing, the component labels are independent fair choices. Hence conditioning on either phase does not change the law of the uncolored pairing or the spacings. For the empty diagram, the same conclusion follows from the auxiliary fair phase.

Suppose the largest component leaves $r$ other chords. Let its successive endpoint gaps contain $q_1,\ldots,q_m$ elementary spacings. Every endpoint not in this component belongs to one of the $r$ remaining chords, so
\be
\sum_jq_j=2k,\qquad q_j\leq2r+1.
\ee
A test chord avoiding a connected component must have both endpoints in the same gap between consecutive endpoints of the component. Indeed, otherwise the two endpoints of the test chord divide the component's endpoints into two nonempty sets. If no chord crosses the test chord, the component splits into two collections with no crossing between them, contradicting connectivity. Conversely, a chord with both endpoints in one gap crosses none of the component's chords.

If the physical gap lengths are $g_j$, the avoidance probability is therefore $\sum_jg_j^2$. The Dirichlet second moments give, for a fixed rooted pairing $\pi$,
\be
\mathbb E\left[\sum_jg_j^2\,\middle|\,\pi\right]
=\frac{\sum_jq_j(q_j+1)}{2k(2k+1)}
\leq\frac{2r+2}{2k+1}.
\label{eq:glass_gap_bound}
\ee
Recall that $\phi$ is the probability that a fresh chord avoids every onsite chord of the selected family. Its conditional mean is bounded by the left-hand side of \cref{eq:glass_gap_bound}. Tightness of $r$ now implies that this mean tends to zero as $k\to\infty$. For any fixed $R$, split into the events $r\leq R$ and $r>R$. The first contributes at most $(2R+2)/(2k+1)$, and the second contributes at most its probability, since $0\leq\phi\leq1$. First choose $R$ large and then $k$ large.

\subsection{Averaging Over the Chord Count}

For a nonempty rooted pairing, put
\be
\psi(\pi)=\frac{\sum_jq_j(q_j+1)}{2k(2k+1)},
\ee
where the gaps are those of the selected largest component. Put $\psi=1$ for the empty diagram. This is the spacing-averaged probability of avoiding that component, so $0\leq\psi\leq1$. Define
\be
\delta_b=
\frac1{F_b}\sum_{k\geq0}\frac{(b^2/2)^k}{k!}
\mathbb E_{\pi}\left[2^{C_k}\psi(\pi)\right],
\label{eq:glass_delta}
\ee
where $\pi$ is a uniform rooted pairing. Since conditioning on the phase leaves the weighted pairing and spacing laws unchanged, the preceding bound gives
\[
\mathbb E_{\mu_b^\sigma}\phi\leq\delta_b
\]
for either phase $\sigma$.

Write $a=b^2/2$. The chord count has weights proportional to $a^kw_k/k!$, where $1\leq w_k\leq B_0$ and $F_b\geq e^a$. Thus, for every fixed $m$,
\[
\Pr(k\leq m)\leq B_0e^{-a}\sum_{k=0}^m\frac{a^k}{k!}
\longrightarrow0.
\]
The weighted mean of $\psi$ tends to zero with $k$, and $0\leq\psi\leq1$. Hence $\delta_b\to0$.

This proves \cref{lem:glass_onsite}.

\section{A Finite-Size Upper Bound}
\label{app:glass_finite}

The physical chord expansion assigns a commutation sign to each crossing of opposite families. Averaging the microscopic labels nearly cancels these signs. We compare the resulting weights with $F_G$ to prove \cref{lem:glass_finite_upper}.

\subsection{The Paired Expansion}

Expand $\mathcal Z_{n,G}$ in the individual Pauli terms of \cref{eq:glass_block,eq:glass_coupling}; we call these terms the primitive operators. Redundancy-freeness removes every word in which some primitive appears an odd number of times. For each remaining word, pair the appearances of each primitive and join paired positions by chords. Bringing the two operators in each pair together gives one commutation sign for each crossing of chords. Since every primitive squares to the identity, the product of these signs is exactly the normalized trace of the word.

Four appearances of the same primitive admit three pairings. More generally, $2m_p$ appearances of a primitive $p$ admit $(2m_p-1)!!$ pairings, so each pairing needs a compensating factor.

An onsite species specifies a block and one of its two families; an edge species specifies an edge and one of its two types. Each species contains $n$ primitive operators, indexed by $\mathbb F_n$. Assign each chord an independent uniform microscopic label in this field. For a diagram $\mathcal D$ with labels $\mathbf x$, let $m_p$ be the number of chords assigned to primitive $p$, and put
\be
\Xi(\mathcal D,\mathbf x)=
\prod_p\frac1{(2m_p-1)!!},
\qquad (-1)!!=1.
\label{eq:glass_multiplicity}
\ee
Let $\chi_A(\mathcal D,\mathbf x)$ be the product of the actual Pauli crossing signs, including those on the edge ancillas. Then
\be
\mathcal Z_{n,G}=\int w_n(\mathcal D)\,d\nu(\mathcal D),
\qquad
w_n(\mathcal D)=\mathbb E_{\mathbf x}
[\Xi(\mathcal D,\mathbf x)\chi_A(\mathcal D,\mathbf x)].
\label{eq:glass_exact_chords}
\ee
Here $\nu$ is the same product measure as in \cref{eq:glass_chord_measure,eq:glass_positive_model}.

To check the coefficients, a species whose primitive has dimensionless coupling $h$ contributes $nh^2/2$ per chord: $n$ choices of microscopic label and a factor $1/2$ for the two orders of its endpoints. This gives intensity $b^2/2$ for an onsite species and $t^2/2=K$ for an edge species. The factorials in the chord measure account for permuting chords of the same species, while integrating their endpoints accounts for the orders of appearances in the exponential series. Multiplying by the correction in \cref{eq:glass_multiplicity} therefore recovers each even word with its correct coefficient.

Since $0<\Xi\leq1$ and $\nu$ has total mass $\exp(Nb^2+2KM)$, these rearrangements are absolutely convergent. The positive model replaces $w_n$ by $\prod_v I_v$.

\subsection{The Binary Matrices}

Write $n=2^r$, with $r\geq3$, and let $\mathbb F=\mathbb F_{2^r}$. For $\alpha\neq0$, define the binary matrices
\be
M_\alpha(x,y)=\operatorname{tr}_{\mathbb F}(\alpha xy),
\qquad A_\alpha=I+M_\alpha,
\ee
where the field trace takes values in $\mathbb F_2$. Both matrices are symmetric, and $M_\alpha M_\gamma=0$ over $\mathbb F_2$. Indeed, an entry of this product is the parity of the number of $y$ for which two binary linear functionals of $y$ are both one. This number is either zero, $2^{r-1}$, or $2^{r-2}$, all even for $r\geq3$. It follows that $A_\alpha^2=I$.

For $x\neq0$, the functional $y\mapsto\operatorname{tr}_{\mathbb F}(\alpha xy)$ is nonzero, so its row has $n/2$ ones. Adding the identity toggles one entry, giving row weight $n/2\pm1$. The row labeled by zero has weight one.

All nonzero $\alpha$ give permutation-equivalent blocks. If $c^2=\alpha$, which has a unique solution in $\mathbb F$, then
\be
A_\alpha(x,y)=A_1(cx,cy).
\ee
There is also an exact exchange of the two onsite families. Let $U_A|z\rangle=|Az\rangle$ on the binary computational basis, and let $\mathcal H$ apply a Hadamard to every qubit. Since $A=A^{\mathsf T}=A^{-1}$, the Clifford $\mathcal H U_A$ sends $X_i$ to $D_i$ and $D_i$ to $X_i$.

Put $\chi(z)=(-1)^{\operatorname{tr}_{\mathbb F}(z)}$. Away from coincidences between incident chord labels, the crossing sign is
\be
\chi\left(\sum_{i<j}c_{ij}x_ix_j\right),
\label{eq:glass_quadratic_sign}
\ee
where $c_{ij}$ is the sum of $\alpha_v$ over the blocks at which chords $i,j$ cross with opposite family labels. A pair of chords meets at at most two blocks. If it meets at two, both chords lie on the same logical edge $uv$. Thus an incompatible crossing has coefficient either $\alpha_v$ or $\alpha_u+\alpha_v$. Both are nonzero, since adjacent blocks were assigned different nonzero field elements.

Suppose we fix all microscopic labels except those of one incompatible crossing. The sign then depends on the remaining labels $x,y$ through $cxy+ax+by$, up to an additive constant, with $c\neq0$. For arbitrary $a,b\in\mathbb F$,
\be
\frac1{n^2}\sum_{x,y\in\mathbb F}\chi(cxy+ax+by)
=\frac1n\chi(ab/c).
\label{eq:glass_character}
\ee
Indeed, summing over $x$ gives zero unless $cy+a=0$. There is exactly one such $y$, and substituting it gives the displayed answer.

The identity part of $A_\alpha$ changes this character only when incident labels coincide. Ancilla signs also require coincident labels.

\subsection{Cancellation at Separated Vertices}

For a diagram $\mathcal D$, let $N_G[v]$ denote the closed neighborhood of $v$ in the logical graph, and define
\be
\begin{split}
S(\mathcal D)&=\{v:I_v(\mathcal D)=0\},\\
R_v&=\#\{\hbox{chords incident to }v\},\qquad
L_v=3+4\sum_{u\in N_G[v]}R_u.
\end{split}
\label{eq:glass_local_counts}
\ee
At $v\in S(\mathcal D)$, select two chords crossing with opposite family labels, using a rule independent of their microscopic labels. They belong to different species, and their block supports lie in $N_G[v]$. Condition on the labels of every other chord, which we call the background, and write $x,y$ for the two selected labels.

First suppose neither selected label coincides with the other or with an incident background label. Each selected primitive then appears on only one chord, so its multiplicity factor in \cref{eq:glass_multiplicity} is one. All identity and ancilla corrections involving the selected chords also vanish. The dependence on $x,y$ is therefore
\[
\chi(cxy+ax+by+q),\qquad c\neq0,
\]
multiplied by the background's actual sign and multiplicity factor. The background multiplier has absolute value at most one.

Use this expression as a comparison function for all $x,y$, including coincident labels. By \cref{eq:glass_character}, its unrestricted average has absolute value at most $1/n$. A union bound shows that a coincidence involving a selected label has probability at most
\be
\frac{1+2\sum_{u\in N_G[v]}R_u}{n}.
\ee
The first term accounts for $x=y$; the remaining terms bound coincidences with background chords. On this event, both the exact and comparison expressions have absolute value at most one, so they differ by at most two. Adding this error to the character bound gives
\be
\left|\mathbb E_{x,y}
[\Xi(\mathcal D,\mathbf x)\chi_A(\mathcal D,\mathbf x)
\mid\hbox{background}]\right|
\leq\frac{1+2+4\sum_{u\in N_G[v]}R_u}{n}
=\frac{L_v}{n}.
\label{eq:glass_one_pair}
\ee
Adding the selected chords can only increase the double-factorial denominators in \cref{eq:glass_multiplicity}. Thus, even on a coincidence, the exact absolute weight is at most the background's multiplicity weight. The same argument also bounds the conditional average by this background weight times $L_v/n$.

We can multiply these estimates when the selected vertices have disjoint closed neighborhoods. Their selected pairs then have disjoint block supports. After conditioning on all other labels, no crossing sign or multiplicity factor couples two selected pairs, so the averages over the pairs factor. Consequently, for every such set $T\subseteq S(\mathcal D)$,
\be
|w_n(\mathcal D)|\leq\prod_{v\in T}\frac{L_v}{n}.
\ee
The square of $G$, which joins vertices at distance one or two, has maximum degree at most $3+3\cdot2=9$. Color it with ten colors. Vertices of one color have disjoint closed neighborhoods. Apply the preceding bound to each color class of $S(\mathcal D)$ and take the geometric mean of the ten bounds. We obtain
\be
|w_n(\mathcal D)|
\leq n^{-|S(\mathcal D)|/10}
\prod_{v\in S(\mathcal D)}L_v^{1/10}.
\label{eq:glass_all_pairs}
\ee
\subsection{Summing the Incompatible Diagrams}

First consider a compatible diagram, so that $S(\mathcal D)=\emptyset$. Every crossing pair uses the same family at each shared block and hence commutes. The ancillas cannot add a sign either: a crossing with opposite ancilla axes would come from the two opposite terms on one edge, already violating compatibility at its endpoint blocks. Therefore $\chi_A=1$ for every microscopic label assignment, and
\be
0\leq w_n(\mathcal D)\leq1.
\ee
Thus the total contribution of compatible diagrams is at most $F_G$.

For the remaining diagrams, fix a nonempty incompatible set $S$ and put $W=N_G[S]$, the union of its closed neighborhoods. Since the maximum degree is three, $|W|\leq4|S|$. By \cref{eq:glass_all_pairs}, dropping the requirement of incompatibility inside $S$ while retaining compatibility outside it gives
\be
\int_{S(\mathcal D)=S}|w_n(\mathcal D)|\,d\nu
\leq n^{-|S|/10}
\int\prod_{v\in S}L_v^{1/10}
\prod_{v\notin S}I_v(\mathcal D)\,d\nu.
\label{eq:glass_fixed_S}
\ee
The factors $L_v$ depend only on chords incident to $W$. We therefore integrate without constraints every species touching $W$: the onsite species there and the edge species on edges incident to $W$. We also drop every compatibility constraint involving one of these chords. The remaining constraints then involve only the induced graph on $V\setminus W$.

Put $a=b^2+6K$. Their total intensity is
\be
A_W=b^2|W|+2K\,\#\{e:e\cap W\neq\emptyset\}
\leq a|W|\leq4a|S|.
\ee
Here the two onsite species have combined intensity $b^2$, and at most three edges per vertex each contribute combined intensity $2K$.

Let $m_W$ be the total chord count of these species. Each chord is incident to at most two blocks, and each block belongs to at most four closed neighborhoods. Hence
\be
\sum_{v\in S}\sum_{u\in N_G[v]}R_u\leq8m_W.
\ee
Using $3+4x\leq3e^{4x/3}$ in the definition of $L_v$, we get
\be
\prod_{v\in S}L_v^{1/10}
\leq3^{|S|/10}e^{16m_W/15}.
\ee
For the unnormalized chord measure with total intensity $A$,
\[
\int e^{cm}\,d\nu=\sum_{m\geq0}\frac{A^m e^{cm}}{m!}
=e^{Ae^c}.
\]
Applying this with $c=16/15$ bounds the removed-species integral by $C_{b,K}^{|S|}$, where we may take
\be
C_{b,K}=3^{1/10}
\exp\left[4(b^2+6K)e^{16/15}\right].
\label{eq:glass_finite_constant}
\ee
After integrating the removed species, the remaining integral is the positive partition function on $G[V\setminus W]$. It is at most $F_G$: the term with zero chords in every removed species contributes exactly this smaller-graph integral to $F_G$, and all other contributions to $F_G$ are nonnegative. Thus \cref{eq:glass_fixed_S} becomes
\be
\int_{S(\mathcal D)=S}|w_n(\mathcal D)|\,d\nu
\leq F_G\left(C_{b,K}n^{-1/10}\right)^{|S|}.
\ee
Summing over nonempty $S$ and adding the compatible contribution gives
\be
\mathcal Z_{n,G}
\leq F_G\sum_{S\subseteq V}
\left(C_{b,K}n^{-1/10}\right)^{|S|}
=F_G\left(1+C_{b,K}n^{-1/10}\right)^N.
\ee
This proves \cref{lem:glass_finite_upper}.

\bibliography{rfree-ref}
\end{document}